\documentclass [10pt]{article}

\usepackage[utf8]{inputenc}
\usepackage[T1]{fontenc} 
\usepackage[french,english]{babel}

\usepackage{amssymb}
\usepackage{amsfonts}
\usepackage{graphicx}
\usepackage{amsmath,empheq}
\usepackage{empheq}

\usepackage{xcolor}
\usepackage{here}
\usepackage{lineno}

\newtheorem{theorem}{Theorem}[section]

\newtheorem{claim}[theorem]{Claim}

\newtheorem{lemma}[theorem]{Lemma}

\newenvironment{proof}[1][Proof]{\noindent\textit{#1.} }{\hfill \rule{0.5em}{0.5em}}

\newcommand{\R}{\mathbb{R}}

\newcommand{\I}{\mathbb{I}}

\renewcommand{\eqref}[1]{{\rm(\ref{#1})}}
\renewcommand{\d}{{\rm d}}

\usepackage{changes}

\definechangesauthor[name={Samuel Alizon}, color=red]{SA}
\definechangesauthor[name={Mircea Sofonea}, color=blue]{MS}
\definechangesauthor[name={RamsesDD}, color=orange]{RDD}

\title{\textsc 
Within-host bacterial growth dynamics with both mutation and horizontal gene transfer
}
\author{Ramsès Djidjou-Demasse$^{a,*}$, Samuel Alizon$^{a}$, Mircea T. Sofonea$^{a}$\\
	{\small $^{a}$ MIVEGEC, IRD, CNRS, Univ. Montpellier, Montpellier, France
	\footnote{Author for correspondence: ramses.djidjoudemasse@ird.fr}}
}

\begin{document}
\maketitle
\begin{abstract}
The evolution and emergence of antibiotic resistance is a major public health concern. The understanding of the within-host microbial dynamics combining  mutational processes, horizontal gene transfer and resource consumption, is one of the keys to solve this problem. We analyze a generic model to rigorously describe interactions dynamics of four bacterial strains: one fully sensitive to the drug, one with mutational resistance only, one with plasmidic resistance only and one with both resistances. By defining thresholds numbers (i.e.~each strain’s effective reproduction and each strain’s transition thresholds numbers), we first express conditions for the existence of non trivial stationary states. We find that these thresholds mainly depend on bacteria quantitative traits such as nutrient consumption ability, growth conversion factor, death rate, mutation (forward or reverse) and segregational loss of plasmid probabilities (for plasmid-bearing strains).  Next, with respect to the order in the set of  strain’s effective reproduction thresholds numbers, we show that the qualitative dynamics of the model range from the extinction of all strains, coexistence of sensitive and mutational resistance strains to the coexistence of all strains at equilibrium. Finally, we go through some applications of our general analysis depending on whether bacteria strains interact without or with drug action (either cytostatic or cytotoxic). 
\end{abstract}

	
\section{Introduction}

Antibiotic, or antibacterial, resistance (ABR) is a major public health concern worlwide. In the USA, for instance, it has been estimated that, each year, two million people suffer infections from antibiotic resistant bacteria, 35,000 of which lead to death \cite{cdc2019}. As any trait under natural selection, antibiotic resistance requires that some individual bacteria differ from the sensitive ('wild type') genotype. In short, these individuals can emerge due to mutations, which here encompass any genetic change in the bacterial chromosome that lead to resistance, but also through horizontal gene transfer (HGT). Indeed, many bacterial species can transmit plasmids through conjugation \cite{Thomas2005}, which are self-replicating DNA often encoding drug resistance genes \cite{Philippon2002,Yates2006, Martinez2008,Moellering2010,Bush2011,Carattoli2012}. The origin of antibiotic resistance, through mutation or HGT is important from both epidemiological and evolutionary standpoint. Combination of modern techniques such as whole genome sequencing and phylogenetic reconstruction have for example revealed the evolutionary history of a {\it Staphylococcus aureus} lineage  that spread in the UK during the 1990s. It turns out that this lineage had acquired resistance both to meticillin and ciprofloxacin -- which belong to two distinct antibiotic classes recently introduced into clinical practice -- respectively through HGT and point mutation \cite{Baker2018}.

Several models have been developed to study the evolution of resistance by one or the other mechanism \cite{Blanquart2019} but few consider the two processes. There are exceptions and, for instance, Tazzyman and Bonhoeffer studied the difference of chromosomal and plasmid mutation in an emergence context, where stochasticity is strong \cite{Tazzyman2014}. Svara and Rankin \cite{Svara2011} develop such a setting to study the selective pressures that favour plasmid-carried antibiotic resistance genes. They use mathematical models capturing plasmid dynamics in response to different antibiotic treatment regimes. Here, we adopt a more formal approach than Svara and Rankin \cite{Svara2011} to investigate the properties of the dynamical system describing interactions between bacteria strains. More precisely, we describe qualitatively within-host interactions dynamics, while taking into account the main quantitative traits of the bacteria life cycle introduced below.  

In this work, we use a system of ordinary differential equations to model the interaction dynamics between four bacterial strains that are fully sensitive to the drug ($N_s$), with mutational (or `genomic') resistance only ($N_m$), with plasmidic resistance only ($N_p$) and with both form of resistances ($N_{m.p}$). Here, $N_j$'s refer to population size and  we  denote as $\mathcal{J}=\left\{s,m,p,m.p\right\}$ the set of bacteria strains. Each $N_j$-strain is then characterized by a strain-specific nutrient consumption ability, growth conversion factor, death rate and mutation (forward or reverse) probability. Additionally, $N_p$ and $N_{m.p}$-strains can transmit plasmids  through HGT and an offspring from plasmid-bearing strain can be plasmid‐free with some probability, a biological process known as segregation.

This paper is organized as follows.  In Section \ref{sec-GenModel} we describe the original model developed in this work. Section \ref{sec-threshold-dynamic} is devoted to some general remarks on the model properties and the threshold asymptotic dynamics. Section \ref{sec-Equi-Gen-Model} investigates the existence of nontrivial stationary states of the model as well as their stability. In particular, we show that the dynamical behavior of the system is not trivial and can range from the extinction of all strains to the coexistence of two or all of them. Finally, in Section \ref{sec-discuss}, we discuss some scenarios that can be captured by the model, as well as the biological implications of model assumptions and limitations.


\section{The model description}	\label{sec-GenModel}
\begin{figure}[!htp]
	\centering
	\includegraphics[width=.75\textwidth]{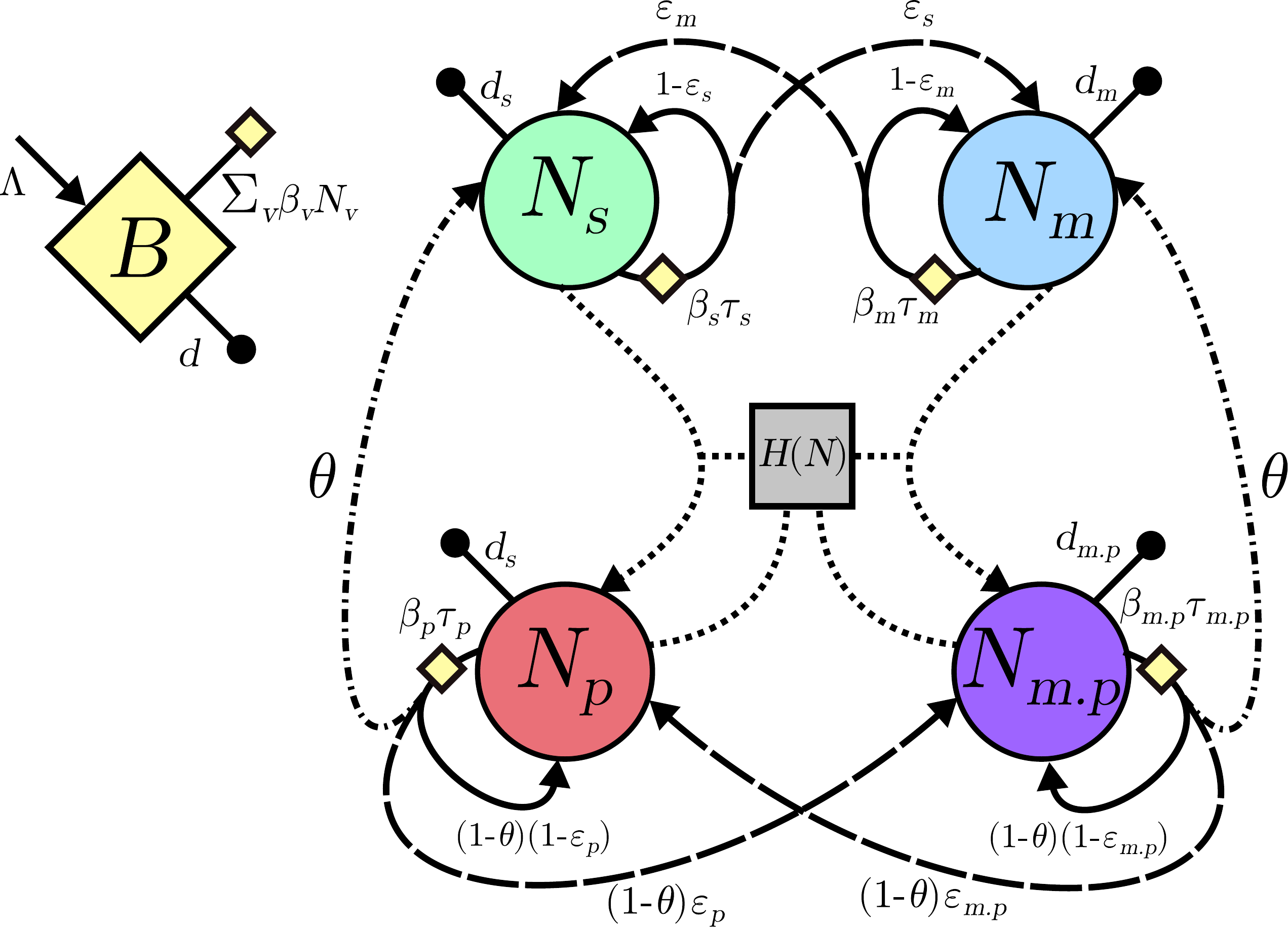}
	\caption{Flow diagram illustrating the interactions dynamics followed by the four bacterial strain densities: drug-sensitive $N_s$, mutation-acquired resistant $N_m$, horizontally-acquired resistant $N_p$ and doubly resistant $N_{m.p}$, growing on a limiting nutrient concentration $B$. The nutrient is renewed by a constant inflow $\Lambda$ and depleted independently from the focal bacteria at a rate $d$. Bacterial strain $v\in\left\{s;m;p;m,p\right\}$ consumes the nutrient at rate $\beta_v$, which is converted into growth rate through coefficient $\tau_v$. Strains $N_s$ and $N_p$ mutate at rates $\varepsilon_s$ and $\varepsilon_p$ respectively while reverse mutation occur to strains $N_m$ and $N_{m.p}$ at rates $\varepsilon_m$ and $\varepsilon_{m.p}$ respectively (dashed lines). Segregational loss of plasmid occurs during cell growth at rate $\theta$ for strains bearing a plasmid $N_p$ and $N_{m.p}$ (dotted-dashed lines).  Horizontal gene transfer occur from the existing plasmid bearing population at rate $H(N) [N_{m.p}+N_p]$, where $H(\cdot)$ is the Beddington-DeAngelis functional response defined by $H(N)=\frac{\alpha}{a_H+b_HN},$ and wherein $\alpha$ is the {\it flux rate} of HGT (dotted lines). This function covers different mechanisms for HGT: (i) density-dependent, {\it i.e.} HGT rate is proportional to the density of the donor; $b_H= 0; a_H=1$, (ii) frequency-dependent, {\it i.e.} HGT rate is proportional to the frequency of the donor; $a_H=0; b_H=1$ or (iii) a mixed between density- and frequency-dependent; $a_H>0;b_H>0$. 
	}
	\label{ModelFig}
\end{figure}

We formulate a general within-host model for interactions between four bacteria strains that are (i) fully sensitive to the drug $N_s$, (ii) with mutational (or `genomic') resistance only $N_m$, (iii) with plasmidic resistance only $N_p$ and (iv) with both resistances $N_{m.p}$. Figure \ref{ModelFig} summarizes interactions between the four compartments of bacteria strains. We assume that the dynamics of the total bacteria population in the host $N=N_s+N_p+N_m+N_{m.p}$ reads 
\begin{equation}\label{eq-totalPop}
\dot N(t)= \sum_{j \in \mathcal{J}}\left( \mu_j B(t) -d_j\right)  N_j(t),
\end{equation}
where $B(t)$ denotes the concentration of nutrients available at time $t$ for all stains. Parameter $\mu_j$ mimics the capacity of a given $N_j$-strain to take advantage from the resource, while $d_j$ is its death rate. 

In the literature, the occurrence of new mutants depends either on (i) the abundance of the parental cells or (ii) both the abundance and growth rate of the parental cells \cite{Loewe5650,Sniegowski2004,Zur2010}. We first assume the latter (the model under assumption (i) will be discussed later). Therefore, the dynamics of interactions between bacteria strains is given by 

\begin{equation}\label{eq-GenModel}
\left\{
\begin{aligned}
\dot B(t)=& \Lambda-dB -B\sum_{j \in \mathcal{J}}\beta_jN_j ,\\
\dot N_s(t)=&\tau_s\beta_s(1-\varepsilon_s) BN_s  -d_sN_s +\varepsilon_m\tau_m\beta_m BN_m + \theta\tau_p\beta_p BN_p\\
&- H(N) N_s[N_p +N_{m.p}] ,\\
\dot N_m(t)=& \tau_m\beta_m (1-\varepsilon_m) BN_m  -d_mN_m + \varepsilon_s\tau_s\beta_sB N_s +\theta\tau_{m.p}\beta_{m.p} BN_{m.p}   \\& - H(N) N_m[N_{m.p}+N_p],\\
\dot N_p(t)=& \tau_p\beta_p(1-\theta)(1-\varepsilon_p) BN_p  -d_pN_p +\varepsilon_{m.p}(1-\theta)\beta_{m.p}\tau_{m.p}B N_{m.p}\\ &  +H(N) N_s[N_p +N_{m.p}] ,\\
\dot N_{m.p}(t)=& \tau_{m.p}\beta_{m.p}(1-\theta)(1-\varepsilon_{m.p}) BN_{m.p}  - d_{m.p}N_{m.p}  + \varepsilon_p(1-\theta)\tau_p\beta_pBN_p\\ & + H(N) N_m[N_{m.p}+N_p],
\end{aligned}
\right.
\end{equation} 
coupled with initial condition $N_j(0)\ge0$ for all $j\in \mathcal{J}$ and $B(0)\ge 0$. 

More precisely, parameters of System \eqref{eq-GenModel} are summarized in Table \ref{Tab-ModelParameters} and are such that
\begin{itemize}
	\item Nutrients are produced at a constant rate $\Lambda$ and washed out at rate $d$. $\beta_j$ is the yield constant representing the amount of nutrients taken by the specific $N_j$-strain. Then, the $N_j$-strain takes advantage from nutrient at rate $\mu_j=\beta_j \tau_j$; wherein $\tau_j$ is the conversion rate (how much growth is obtained from a unit of nutrient). 
	\item Upon cell replication, sensitive ($N_s$) and resistant with plasmid ($N_p$) strains acquire mutations at ratios $\varepsilon_s$ and $\varepsilon_p$ respectively. Reverse mutation from mutational resistant strain ($N_m$) and mutational with plasmid ($N_{m.p}$) also occur upon replication at ratios $\varepsilon_m$ and $\varepsilon_{m.p}$ respectively. 
	\item The segregational loss of plasmid occurs during cell growth at ratio $\theta$ for the resistant strain with plasmid ($N_p$) and mutational with plasmid ($N_{m.p}$). Here, we assume that the segregational loss of plasmid is the same for all strains bearing a plasmid ($N_p$ or $N_{m.P}$). 
	\item  We assume that plasmids are acquired from the existing donor population ($N_p+N_{m.p}$) and that conjugation follows a mass action law; i.e.~a plasmid is acquired from the existing resistant population at rate $H(N) [N_{m.p}+N_p]$, where $H(\cdot)$ is the Beddington-DeAngelis functional response defined by $H(N)=\frac{\alpha}{a_H+b_HN},$ with $\alpha$  the {\it flux rate} of HGT. This function can capture different mechanisms for HGT: (i) density-dependent, {\it i.e.}~HGT rate is proportional to the density of the donor; $b_H= 0; a_H=1$, (ii) frequency-dependent, {\it i.e.}~HGT rate is proportional to the frequency of the donor; $a_H=0; b_H=1$ or (iii) a mixed between density- and frequency-dependent; $a_H>0;b_H>0$. Throughout this work, we assume that the flux rate $\alpha>0$ for HGT is a small parameter \cite{Levin1979,Donisio2002,Gordon1992}.
\end{itemize}


\begin{table}[H]
	\begin{tabular}{lll}
		\hline \hline
		State variables & \multicolumn{2}{l}{Description}\\
		\hline \hline
		$B(t)$ &  \multicolumn{2}{l}{Density of nutrients at time $t$.}\\
		$N_s(t)$ &  \multicolumn{2}{l}{Density of sensitive strain at time $t$.}\\
		$N_m(t)$ &  \multicolumn{2}{l}{Density of resistant strain at time $t$}.\\ 
		& \multicolumn{2}{l}{ through mutational changes.} \\
		$N_{p}(t)$ &  \multicolumn{2}{l}{Density of resistant strain at time $t$}\\
		& \multicolumn{2}{l}{through Horizontal Gene Transfer (HGT).}\\
		$N_{m.p}(t)$ &  \multicolumn{2}{l}{Density of resistant strain at time $t$.}\\ 
		& \multicolumn{2}{l}{through mutational changes and HGT.}\\ 
		$N(t)$ &  \multicolumn{2}{l}{Total density of bacteria at time $t$:}\\
		& \multicolumn{2}{l}{ $N=N_s+N_p+N_m+N_{m.p}$.}\\
		\hline \hline
		Model parameters & Description (unit) & Value/range [ref.]\\
		\hline \hline
		$\varepsilon_v$ & Resistance-related mutation proportion & $10^{-8}$ \cite{ankomah_exploring_2014}\\
		& per cell division of $v$-strain (dimensionless). & \\
		$\beta_v$ & Nutrient consumption by strain $v$ bacteria & $10^{-6}$ \cite{colijn_how_2015}\\
		& (mL/cell/day).  & \\
		$\tau_v$ & strain $v$ growth conversion factor (cell/$\mu$g). & $10^6$ \cite{colijn_how_2015}\\
		$d_v$ & strain $v$ washout and death rate (1/day). & 5 \cite{ankomah_exploring_2014}\\
		$d$ & nutrient washout rate (1/day). & 5 \cite{ankomah_exploring_2014}\\
		$\Lambda$ & Nutrient renewal rate ($\mu$g/mL/day). & 2500  \cite{ankomah_exploring_2014}\\
		$\theta$& Proportion of segregational plasmid loss & [0,0.25] \cite{bahl_quantification_2004}\\
		& for $p$ and $m.p$-strains (dimensionless).& \\
		$H(\cdot)$ & Functional response of HGT (mL/cell/day) & $10^{-13}$ \cite{lopatkin_antibiotics_2016}\\
		\hline 
	\end{tabular}
	\caption{Model state variables and parameters} \label{Tab-ModelParameters}
\end{table}

System \eqref{eq-GenModel} will be considered under the following natural assumption
\[
\begin{split}
&\Lambda>0\;, d>0\;, d_v> 0\;, \beta_v\ge0\;,  \tau_v\ge 0 \quad{\text{and}} \quad   \varepsilon_j;\theta \in (0,1). 
\end{split}
\]

It is useful to write System \eqref{eq-GenModel} into a more compact form. To that end, we identify the vector $(N_s,N_m)^T$ together with $u$ and $(N_p,N_{m.p})^T$ together with $v$. Here $x^T$ is set for the transpose of a vector or matrix $x$. Then, System \eqref{eq-GenModel} can be rewritten as 
\begin{equation}\label{eq-GenModelCompact}
\left\{
\begin{split}
\dot B(t)=& \Lambda-B \left[d + \left<\beta,(u,v)^T\right>\right],\\
\dot u=& \left[BG-D-h(u,v)\right] u+ BL_pv,\\
\dot v=& \left[B(G_p-L_p)-D_p\right]v+ h(u,v)u,
\end{split}
\right.
\end{equation}
wherein we have formally set $\beta= \left(\beta_s,\beta_m,\beta_{m.p} \right)^T$, $ h(u,v)= H \left(N\right) \left< 1,v \right>$ and the matrices $G$, $G_p$, $D$, $D_p$ and $L_p$ are defined as follows
\[
\begin{split}
& G= 
\left[
\begin{array}{cc}
\tau_s\beta_s(1-\varepsilon_s) & \varepsilon_m\tau_m\beta_m\\
\varepsilon_s\tau_s\beta_s & \tau_m\beta_m(1-\varepsilon_m)
\end{array}
\right], \quad 
G_p= 
\left[
\begin{array}{cc}
\tau_p\beta_p(1-\varepsilon_p(1-\theta)) & \varepsilon_{m.p}(1-\theta)\tau_{m.p}\beta_{m.p} \\
\varepsilon_p(1-\theta)\tau_p\beta_p & \tau_{m.p}\beta_{m.p} (1-\varepsilon_{m.p}(1-\theta)) 
\end{array}
\right],\\
& D=\text{diag}(d_s,d_m), \quad D_p=\text{diag}(d_p,d_{m.p}), \quad L_p=\theta \text{ diag}(\tau_p\beta_p,\tau_{m.p}\beta_{m.p}),
\end{split}
\]
wherein $\text{diag}(w)$ is a diagonal matrix the diagonal elements of which are given by $w$; and $\left<x,y\right>$ is set for the usual scalar product of vectors $x$ and $y$. The compact form \eqref{eq-GenModelCompact} is then view as the inetraction dynamics of two groups of bacteria: one with plasmids $v=\left(N_p,N_{m.p}\right)^T$ and one without it $u=\left(N_s,N_m\right)^T$.



\section{General remarks, bacteria invasion process and threshold dynamics}\label{sec-threshold-dynamic}
In this section we establish some useful properties of solutions of \eqref{eq-GenModel} that include the existence of a positive global in time solution of the system and the threshold asymptotic dynamics of the model.

\subsection{General remarks}
First, in a bacteria-free environment, the nutrient dynamics is such that
\begin{equation}\label{B-dfe}
\dot B=\Lambda-dB.
\end{equation}
Consequently, $B_0= \frac{\Lambda}{d}$ is the unique positive solution of System \eqref{B-dfe}, which can straightforwardly be shown to be globally attractive in $\R_+$; {\it i.e.}~$B(t)\to B_0$ as $t\to \infty$.

Since System \eqref{eq-GenModel} is designed to model a biological process, its solutions should remain positive and bounded. The existence, positivity and boundedness solutions of System \eqref{eq-GenModel} is provided by the following result (see Section \ref{proof of thm-exist} for the proof).
\begin{theorem}\label{thm-exist}
	There exists a unique continuous solution $\left\{E(t):\R_+^5 \to \R_+^5 \right\}_{t\ge0}$ to \eqref{eq-GenModel} such that for any $\omega_0 \in \R^5_+$, the orbit of \eqref{eq-GenModel} passing through $\omega_0$ at time $t=0$ $E(\cdot)\omega_0: [0,\infty) \to \R^5_+$ defined by $E(t)\omega_0= \left( B(t),N_s(t),N_m(t),N_p(t),N_{m.p}(t)\right)$ with $E(0)\omega_0= \omega_0$ verified
	\begin{equation}\label{eq-bounded}
	\tau_{\text{max}} B(t) + \sum_{j\in\mathcal{J}}N_j\le \frac{\Lambda\tau_{\text{max}}} {\min \left(d,d_{\text{min}} \right)},
	\end{equation}
	wherein $\mathcal{J}=\left\{s,m,p,m.p\right\}$; $\tau_{\text{max}}= \max\limits_{j \in \mathcal{J}}\tau_j$ and $d_{\text{min}}= \max\limits_{j \in \mathcal{J}}d_j$.
\end{theorem}

\bigskip

Recalling that $u=(N_s,N_m)^T$, $v=(N_p,N_{m.p})^T$ and setting $K=\frac{\Lambda\tau_{\text{max}}} {\min \left(d,d_{\text{min}} \right)}$, we define
\begin{equation*}\label{Omega}
\begin{split}
& \Omega=\left\{ (B,N_s,N_m,N_p,N_{m.p})\in \R_+^5 \left|
\tau_{\text{max}} B +\sum_{j \in \mathcal{J}} N_j\le K \right.
\right\},\\
&X_0=\left\{(B,N_s,N_m,N_p,N_{m.p})\in \Omega:
\sum_{j \in \mathcal{J}} N_j >0 \right\},\quad  \text{and} \quad  \partial X_0= \Omega\setminus X_0.
\end{split}
\end{equation*}
Then the following lemma holds true
\begin{lemma}\label{lem-invariance-X0}
	Let $w_0$ be a given initial data of System \eqref{eq-GenModel} and set $ E(\cdot)w_0$ the orbit passing through $\omega_0$ at $t=0$. The subsets $X_0$ and $\partial X_0$ are both positively invariant under the map $\left\{ E(t)w_0 : [0,\infty) \to R^5_+ \right\}_{t\ge0}$; in other words every solution of system \eqref{eq-GenModel} with initial value in $X_0$ (respectively in $\partial X_0$) stays in $X_0$ (respectively in $\partial X_0$).
\end{lemma}
\begin{proof}
	The proof of Lemma \ref{lem-invariance-X0} is straightforward. Indeed, let us first recall that $B(t)\le B_0$ for all time $t$. Next, the positivity of the map $E(t)w_0$ provided by Theorem \ref{thm-exist} together with equation \eqref{eq-totalPop} give 
	\[
	-d_{\text{max}} N(t) \le \dot N(t) \le B_0 \tau_{\text{max}} \beta_{\text{max}} N(t), \text{ for all time } t.
	\]
	From where $N(0)\exp(-d_{\text{max}} t) \le N(t) \le N(0) \exp(B_0 \tau_{\text{max}} \beta_{\text{max}} t)$, and the result follows.
\end{proof}

\subsection{Bacteria invasion process and threshold dynamics} \label{sec-Inv}
Trivially, the bacteria-free stationary state is given by $E^0= \left(B_0,0,0,0,0\right)$. Recall that each $N_j$-strain is basically characterized by a specific nutrient consumption ability $\beta_j$, growth conversion factor $\tau_j$, death rate $d_j$, and mutation (forward or reverse) probability $\varepsilon_j$. Additionally, $N_p$- and $N_{m.p}$-strains can transmit plasmids  througth HGT and there is generally some probability $\theta$ that any offspring from plasmid bearing strain will be plasmid‐free. The qualitative dynamics of our model is strongly related to the following threshold numbers 
\begin{equation} \label{eq-threshold-Tj}
\mathcal{T}_j= B_0\frac{\tau_j\beta_j} {d_j}; \quad \text{ for } j \in \mathcal{J},
\end{equation}
and 
\begin{equation}\label{eq-threshold-Rj}
\begin{split}
& \mathcal{R}_s= \mathcal{T}_s(1-\varepsilon_s), \quad  \mathcal{R}_m= \mathcal{T}_m(1-\varepsilon_m),\\
&\mathcal{R}_p= \mathcal{T}_p(1-\varepsilon_p)(1-\theta), \quad \mathcal{R}_{m.p}= \mathcal{T}_{m.p}(1-\theta)(1-\varepsilon_{m.p}),
\end{split}
\end{equation}
wherein the positive constant $B_0$ is the bacteria-free stationary state. A threshold $\mathcal{T}_j$ will be referred to below as the growth threshold of the $N_j$-strain in a bacteria-free environment. Indeed, the net growth rate of the $N_j$-strain can be written as $B_0\tau_j\beta_j-d_j= d_j \left(\mathcal{T}_j-1 \right)$. Therefore, the quantity $\left(\mathcal{T}_j-1 \right)$ can be seen as the net growth threshold. Regarding the $\mathcal{R}_j$, since they capture not only the growth but also the genetic component of the dynamics, we call them {\it effective reproduction numbers}. The main biological interpretation here is that the $\mathcal{T}_j$s describe the ecological conditions of invasion while the $\mathcal{R}_j$s describe both the ecological and micro-evolutionary (or, shortly 'eco-evolutionary') conditions of strain persistence, that is their effective genetic contribution to each next generation.

In addition to the $\mathcal{T}_j$'s and $\mathcal{R}_j$'s, transitions between strain state variables, imposed by mutations and plasmids segregation, are of great importance for the dynamical behaviour of the model considered here. Therefore, we similarly introduce parameters (referred to as {\it strain transition numbers})
\begin{equation}\label{eq-threshold-Kj}
\begin{split}
& \mathcal{K}_{s\to m}= B_0\frac{\varepsilon_s\tau_s\beta_s}{d_m}, \quad \mathcal{K}_{m\to s}= B_0\frac{\varepsilon_m\tau_m\beta_m}{d_s}\\
& \mathcal{K}_{p\to m.p}= B_0\frac{\varepsilon_p(1-\theta)\tau_p\beta_p}{d_{m.p}}, \quad \mathcal{K}_{m.p\to p}= B_0\frac{\varepsilon_{m.p}(1-\theta)\tau_{m.p}\beta_{m.p}}{d_p}\\
& \mathcal{K}_{p\to s}= B_0\frac{\theta\tau_p\beta_p}{d_s},\quad  \mathcal{K}_{m.p \to m}= B_0\frac{\theta\tau_{m.p}\beta_{m.p}}{d_m}.
\end{split}
\end{equation}

From the above notations, we can express the threshold criterion for bacteria invasion as follows:
\begin{theorem}\label{Thm-invasion} Assume that mutation rates $\varepsilon_j$ are sufficiently small. Let $E(t)w_0= (B(t),N_s(t),N_m(t),N_p(t),N_{m.p}(t))$ be the orbit of \eqref{eq-GenModel} passing through $\omega_0$ at time $t=0$. Let us set
	\begin{equation}
	\mathcal{T}^*= \max_{j \in \mathcal{J}} \{\mathcal{T}_j\} {\quad \text{and}} \quad \mathcal{R}^*= \max_{j \in \mathcal{J}} \{\mathcal{R}_j\}, 
	\end{equation}
	where thresholds $\mathcal{T}_j$'s and $\mathcal{R}_j$'s are defined by \eqref{eq-threshold-Tj} and \eqref{eq-threshold-Rj}. Then,
	
	\begin{description}
		\item[(i)] 
		The bacteria-free stationary state $E^0$ for System \eqref{eq-GenModel} is locally asymptotically stable if $\mathcal{R}^* <1$ and unstable if $\mathcal{R}^*
		>1$.
		\item[(ii)] 
		The bacteria-free stationary state $E^0$ for System \eqref{eq-GenModel} is globally asymptotically stable if $\mathcal{T}^* <1$.
		\item[(iii)]
		If $\min\limits_{j \in \mathcal{J}}  \mathcal{T}_j^*>1$, then System \eqref{eq-GenModel} is uniformly persistent with respect to the pair $(X_0,\partial X_0)$, in the sense that there exists $\delta>0$, such that for any $w_0 \in X_0$ we have,
		\begin{equation*}\label{persistence1}
		\liminf _{t\to \infty} \Delta (E(t)w_0,\partial X_0) \geq\delta. 
		\end{equation*}
		Here, $\Delta(\cdot,\cdot)$ is set for the semi-distance defined by $\Delta(\mathcal{X},\mathcal{Y})= \sup\limits_{x\in \mathcal{X}}\inf\limits_{y\in \mathcal{Y}}\|x-y\|$. 
	\end{description}
\end{theorem}

Theorem \ref{Thm-invasion} states that bacteria die out when $\mathcal{T}^*<1$ and can potentially persist with the nutrient if $\mathcal{T}^*>1$. Moreover, we always have $\mathcal{R}^* \le \mathcal{T}^* $ and the difference $\mathcal{T}^* - \mathcal{R}^* $ is of order $\max(\varepsilon_s, \varepsilon_m, \theta+ \varepsilon_p, \theta+ \varepsilon_{m.p})$. Therefore, for sufficiently small mutations ratios ($\varepsilon_j$'s) and plasmid loss ($\theta$), we have $\text{sign}(\mathcal{R}^*-1) = \text{sign}(\mathcal{T}^*-1)$. The proof of Theorem \ref{Thm-invasion} is given in Section \ref{proof of Thm-invasion}.


\section{Nontrivial stationary states of Model \eqref{eq-GenModel} and stability results} \label{sec-Equi-Gen-Model}
We have proven that the trivial stationary state $E^0$ is always a solution of System \eqref{eq-GenModel} (see Section \ref{sec-Inv}). However, in some conditions, the model can also have nontrivial stationary solutions. Here, let us recall the strain transition numbers $\mathcal{K}_{s\to m}$, $\mathcal{K}_{m\to s}$, $\mathcal{K}_{p\to m.p}$, $\mathcal{K}_{m.p\to p}$, $\mathcal{K}_{p\to s}$ and $\mathcal{K}_{m.p \to m}$ introduced by \eqref{eq-threshold-Kj}.

\subsection{Stationary states of the general model \eqref{eq-GenModel}}
Model \eqref{eq-GenModel} has three possible stationary states. The extinction of all strains $E^0$, the co-existence of $N_s$ and $N_m$ strains ($E^*_{s-m}$) and the co-existence of all strains ($E^*$). Let $r\left(M\right)$ be the spectral radius of a given matrix $M$. Then, the precise result reads as follows

\begin{theorem}\label{thm-equi-Model}
	\begin{description}
		\item[(i)] If $B_0r\left(D^{-1}G\right)>1$, then the co-existence stationary state of $N_s$ and $N_m$ strains is $E^*_{s-m}= \left(B^*,N^*_s,N^*_m,0,0 \right)$ with 
		\begin{equation*}
		\begin{split}
		& \frac{1}{B^*}= r\left(D^{-1}G\right)= \frac{1}{2B_0} \left\{\mathcal{R}_s +\mathcal{R}_m + \left[\left(\mathcal{R}_s -\mathcal{R}_m \right)^2 +4 \mathcal{K}_{m\to s}  \mathcal{K}_{s\to m} \right]^{1/2}   \right\},\\
		& \left(N^*_s, N^*_m\right) = \frac{d\left(B_0r(D^{-1}G)-1\right)}{\left< (\beta_s,\beta_m),\phi_0\right>}  \phi_0,
		\end{split}
		\end{equation*}
		wherein $\phi_0= \left(\mathcal{R}_s -\mathcal{R}_m + \left[\left(\mathcal{R}_s -\mathcal{R}_m \right)^2 +4 \mathcal{K}_{m\to s}  \mathcal{K}_{s\to m} \right]^{1/2}  ,2 \mathcal{K}_{s\to m}\right)^T$.
		
		\item[(ii)] If $r(D_p^{-1}(G_p-L_p))> r\left(D^{-1}G\right)$ and $B_0r(D_p^{-1}(G_p-L_p)) >1$, then for small flux rate of HGT  $\alpha$, the coexistence of all strains $E^*= \left(B^*,N^*_s,N^*_m,N^*_p,N^*_{m.p} \right)$ is defined by 
		\[
		\begin{split}
		&(N_s^*,N_m^*)= u^0_*+ \alpha u^1_*+ \mathcal{O}(\alpha^2),\\
		&(N_p^*,N_{m.p}^*)= v^0_*+ \alpha v^1_*+ \mathcal{O}(\alpha^2),\\
		&B^*=  b^0_*+ \alpha b^1_*+ \mathcal{O}(\alpha^2),
		\end{split}
		\]
		wherein 
		\[
		\left\{\begin{split}
		& 1/b^0_*=r(D_p^{-1}(G_p-L_p)),\\
		& v^0_*= \frac{  d\left[B_0/b^0_*-1 \right] } { \left< \beta, \left(-(D^{-1}G -1/b^0_*\I)^{-1} D^{-1}L_p\varphi_0, \varphi_0\right) \right> } \varphi_0,\\
		& u^0_* = -(D^{-1}G -1/b^0_*\I)^{-1} D^{-1}L_pv^0_*,\\
		& \varphi_0= \left(\mathcal{R}_p -\mathcal{R}_{m.p} + \left[\left(\mathcal{R}_p -\mathcal{R}_{m.p} \right)^2 +4 \mathcal{K}_{m.p\to p}  \mathcal{K}_{p\to m.p} \right]^{1/2}  ,2 \mathcal{K}_{p\to m.p}\right)^T,
		\end{split}\right.
		\]
	\end{description}
	and $\I$ denotes the identity matrix of the convenient dimension.
\end{theorem}
We referrer to Section \ref{proof of thm-equi-Model} for the proof of Theorem \ref{thm-equi-Model} above.

\subsection{Estimates for small mutation ratios and flux rate of HGT}
Let us denote by $E= (B,N_s, N_m, N_p, N_{m.p})$ the state variables of System \eqref{eq-GenModel}. For a given solution $E$ of System \eqref{eq-GenModel}, we denote by $P_R(E)=(N_m+N_p+N_{m.p})/N $ the resistance frequency or proportion corresponding to this solution. Further, the proportion $P_R(E)$ is decomposed in term of mutational resistance ($\mathcal{P}^*_m=N_m/N$), plasmidic ($\mathcal{P}^*_p=N_p/N$) and mutational/plasmidic ($\mathcal{P}^*_{m.p}=N_{m.p}/N$) such that $P_R(E)= \mathcal{P}^*_m+ \mathcal{P}^*_p+ \mathcal{P}^*_{m.p}.$ Next, mostly based on Theorem \ref{thm-equi-Model} for small mutation ratios $\varepsilon_v$, we derive simple approximations of resistance frequencies $P_R \left( E^*_{s-m}\right)$ and $P_R \left( E^*\right)$ at stationary states $E^*_{s-m}= \left(B^*,N^*_s,N^*_m,0,0 \right)$ and $E^*= \left(B^*,N^*_s,N^*_m,N^*_p,N^*_{m.p} \right)$. For simplicity and without loss of generality, we express/standardize parameters $\epsilon_j$  and $\alpha$ as functions of the same quantity, let us say $\eta$, with $\eta \ll 1$ a small parameter. This parameter $\eta$ can be seen as the overall contribution of genetic processes to the ecological dynamics.

\paragraph{Estimates for $P_R \left( E^*_{s-m}\right)$ for small mutation ratios.} We consider two situations: $\mathcal{R}_s>\mathcal{R}_m$ and $\mathcal{R}_s<\mathcal{R}_m$, which correspond to situations where a mutation provides a fitness cost or a fitness benefit.

{\it First case: $\mathcal{R}_s>\mathcal{R}_m$ {\it i.e.} $\mathcal{T}_s>\mathcal{T}_m$ (mutation providing a fitness cost).} We find, 
\begin{equation}\label{prop Esm Rs is max}
	\left\{
	\begin{split}
	   & P_R \left(E^*_{s-m}\right)=\mathcal{P}^*_m= \frac{d_s}{d_m} \frac{\mathcal{T}_s}{\mathcal{T}_s- \mathcal{T}_m} \eta + \mathcal{O}(\eta^2),\\
	   &\text{with} \quad \mathcal{R}_s= \mathcal{T}_s (1-\varepsilon_s)  >1, \quad \text{and   } \mathcal{T}_s > \mathcal{T}_m.
	\end{split}
	\right.
\end{equation}

{\it Second case: $\mathcal{R}_s<\mathcal{R}_m$, {\it i.e.} $\mathcal{T}_s<\mathcal{T}_m$ (mutation providing a fitness benefit).} We find,
\begin{equation}\label{prop Esm Rm is max}
	\left\{
	\begin{split}
	   & P_R \left(E^*_{s-m}\right)= \mathcal{P}^*_m= 1- \frac{d_m \left(\mathcal{T}_m-\mathcal{T}_s \right) }{4B_0\tau_s\beta_s} \left( \left( \frac{\mathcal{T}_s +\mathcal{T}_m}{\mathcal{T}_s -\mathcal{T}_m} \right)^2  -1\right) \eta  + \mathcal{O}(\eta^2),\\
	   &\text{with} \quad \mathcal{R}_m= \mathcal{T}_m (1-\varepsilon_m)  >1, \quad \text{and   } \mathcal{T}_m > \mathcal{T}_s.
	\end{split}
	\right.
\end{equation}

\paragraph{Estimates of $E^*$ for small mutation ratios and flux rate of HGT.} Here, we also have two situations: in one case the mutation somehow `compensates' for the presence of the plasmid ($\mathcal{R}_p<\mathcal{R}_{m.p}$), whereas in the other mutation and plasmid costs add up ($\mathcal{R}_p>\mathcal{R}_{m.p}$).
\\

{\it First case: $\mathcal{R}_p> \mathcal{R}_{m.p}$, {\it i.e.} $\mathcal{T}_p> \mathcal{T}_{m.p}$ (mutation and plasmid cost add up).} The existence of the stationary state $E^*$ is ensured by
\begin{equation*}
\mathcal{R}_p> \max \left( \mathcal{R}_s, \mathcal{R}_m \right)\quad \text{ and } \quad \mathcal{R}_p >1.
\end{equation*}
Here, the resistance proportion $P_R(E^*)$ is decomposed in term of mutational resistance ($\mathcal{P}^*_m$), plasmidic ($\mathcal{P}^*_p$) and mutational/plasmidic ($\mathcal{P}^*_{m.p}$) such that 
\begin{equation}\label{eq resistance decomposed0}
\begin{split}
&P_R(E^*)= \mathcal{P}^*_m+ \mathcal{P}^*_p+ \mathcal{P}^*_{m.p},\\
&\text{wherein}\\
&	\left\{
	\begin{split}
	   \mathcal{P}^*_m= & \mathcal{O}(\eta), \\
	   \mathcal{P}^*_p= & \frac{1}{\left( 1 + \frac{B_0 \theta} {\mathcal{T}_p (1-\theta)- \mathcal{T}_s}   \frac{B_0\tau_p\beta_p}{d_s} \right)} + \mathcal{O}(\eta),\\
	   \mathcal{P}^*_{m.p}=&  \mathcal{O}(\eta),\\
	   &\text{with} \quad \mathcal{R}_p= \mathcal{T}_p (1-\theta-\varepsilon_p)  >1, \quad \text{and   } \mathcal{T}_p (1-\theta) > \max \left( \mathcal{T}_s, \mathcal{T}_m \right).
	\end{split}
	\right.
\end{split}
\end{equation}

{\it Second case: $\mathcal{R}_p< \mathcal{R}_{m.p}$, {\it i.e.} $\mathcal{T}_p< \mathcal{T}_{m.p}$ (mutation compensates for the presence of the plasmid).} Again, the existence of the stationary state $E^*$ is ensured by 
\begin{equation*}
\mathcal{R}_{m.p}> \max \left( \mathcal{R}_s, \mathcal{R}_m \right)\quad \text{ and } \quad \mathcal{R}_{m.p} >1.
\end{equation*}

Here we have,
\begin{equation}\label{eq resistance decomposed}
\begin{split}
&P_R(E^*)= \mathcal{P}^*_m+ \mathcal{P}^*_p+ \mathcal{P}^*_{m.p},\\
&\text{wherein}\\
&	\left\{
\begin{split}
	\mathcal{P}^*_m=& \frac{\theta B_0^2 \tau_p\beta_p} {\theta B_0^2 \tau_p\beta_p +d_m \left[\mathcal{T}_{m.p}(1-\theta)-\mathcal{T}_m \right] }  + \mathcal{O}(\eta),\\
	\mathcal{P}^*_p= &  \mathcal{O}(\eta),\\
	 \mathcal{P}^*_{m.p}= & \frac{d_m \left[\mathcal{T}_{m.p}(1-\theta)-\mathcal{T}_m \right]} {\theta B_0^2 \tau_p\beta_p +d_m \left[\mathcal{T}_{m.p}(1-\theta)-\mathcal{T}_m \right] } + \mathcal{O}(\eta).
	\end{split}
	\right.
\end{split}
\end{equation}


\subsection{Stability results of the general model \eqref{eq-GenModel}}

\begin{figure}[!htp]
	\centering
	\includegraphics[width=.9\textwidth]{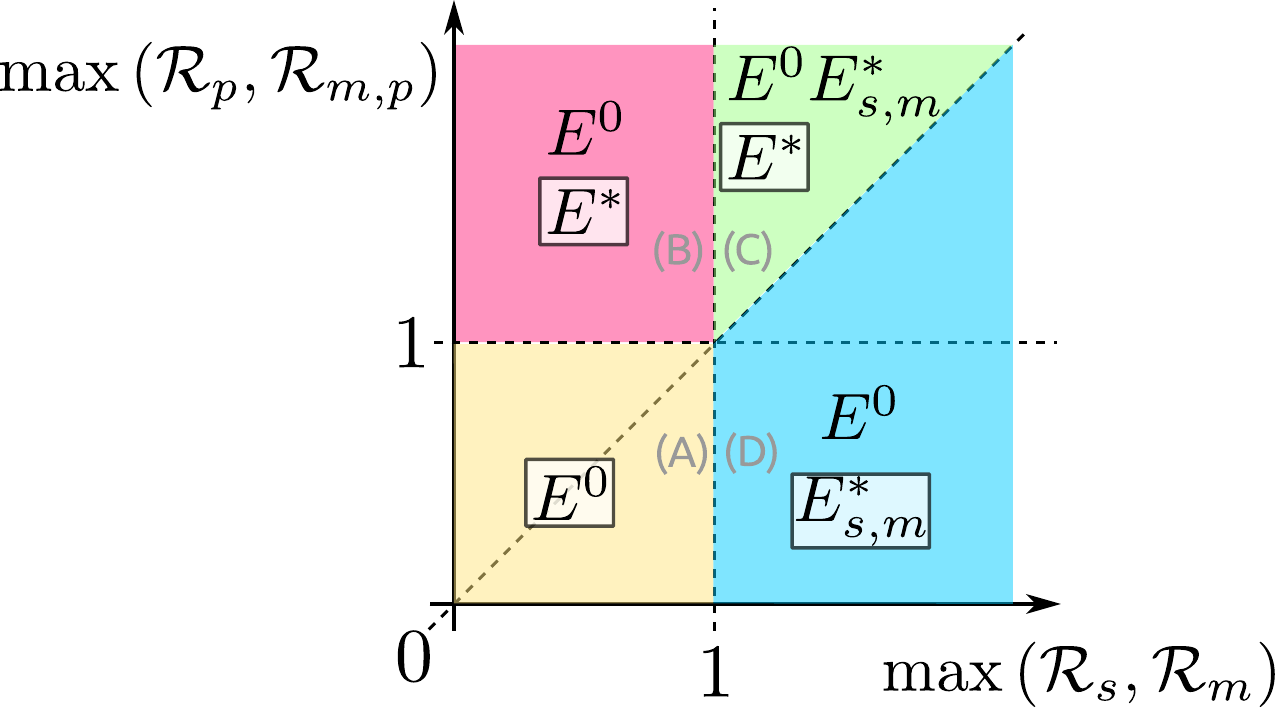}
	\caption{Summary of the qualitative analysis of the general model \eqref{eq-GenModel}. The positive orthant is divided into four zones ((A) to (D)). In each of them, the feasible stationary states are shown -- their notations are boxed if and only if they are locally asymptotically stable (l.a.s.). Note that in zone (D), $E^*_{s,m}$ is l.a.s. providing that \eqref{cond_las_Esm} is satisfied.}
	\label{DiagStability}
\end{figure}

\begin{figure}
	\centering
	\begin{tabular}{cc}
		\includegraphics[width=.5\textwidth]{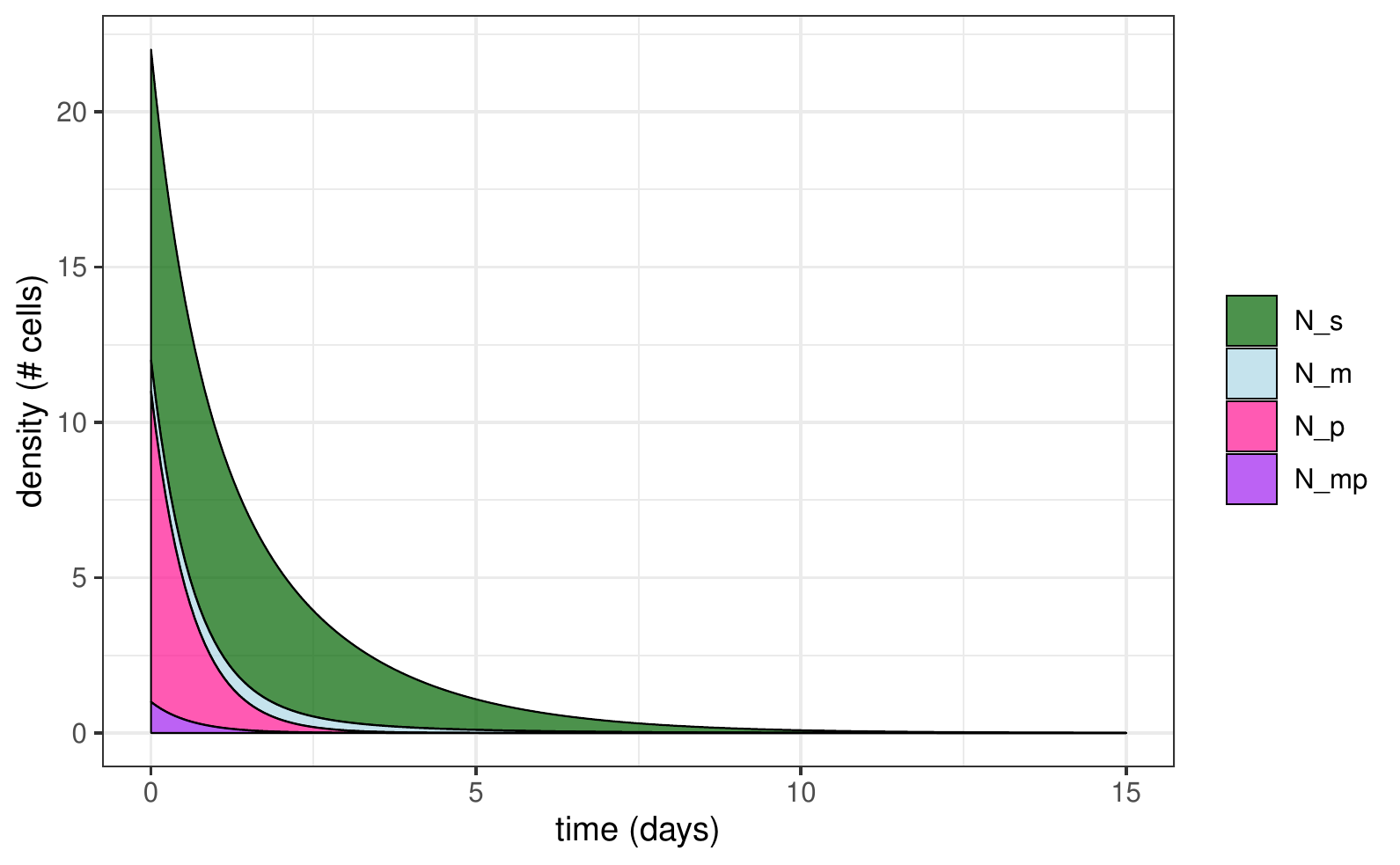} & \includegraphics[width=.5\textwidth]{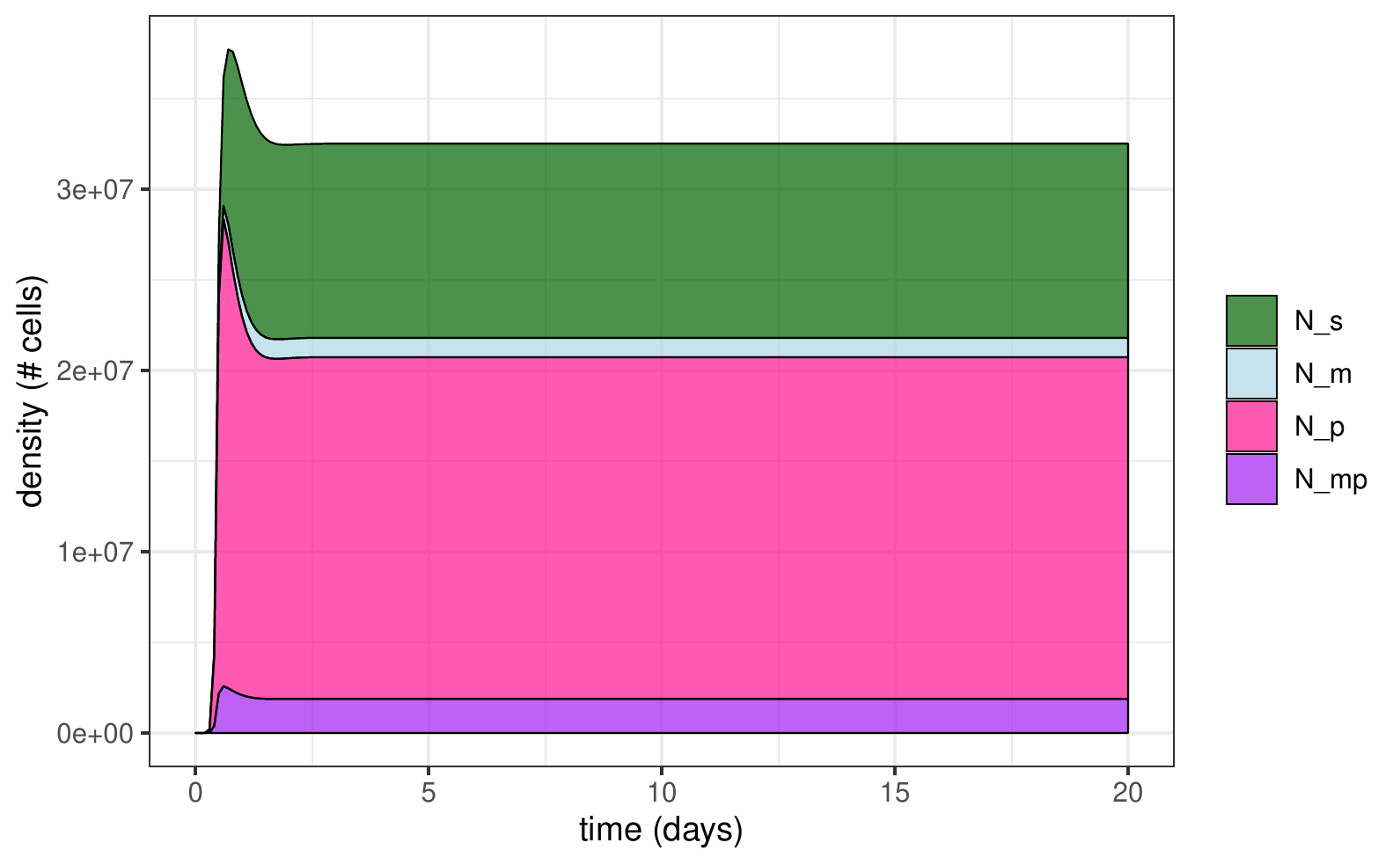}  \\
		(A) & (B)\\
		\includegraphics[width=.5\textwidth]{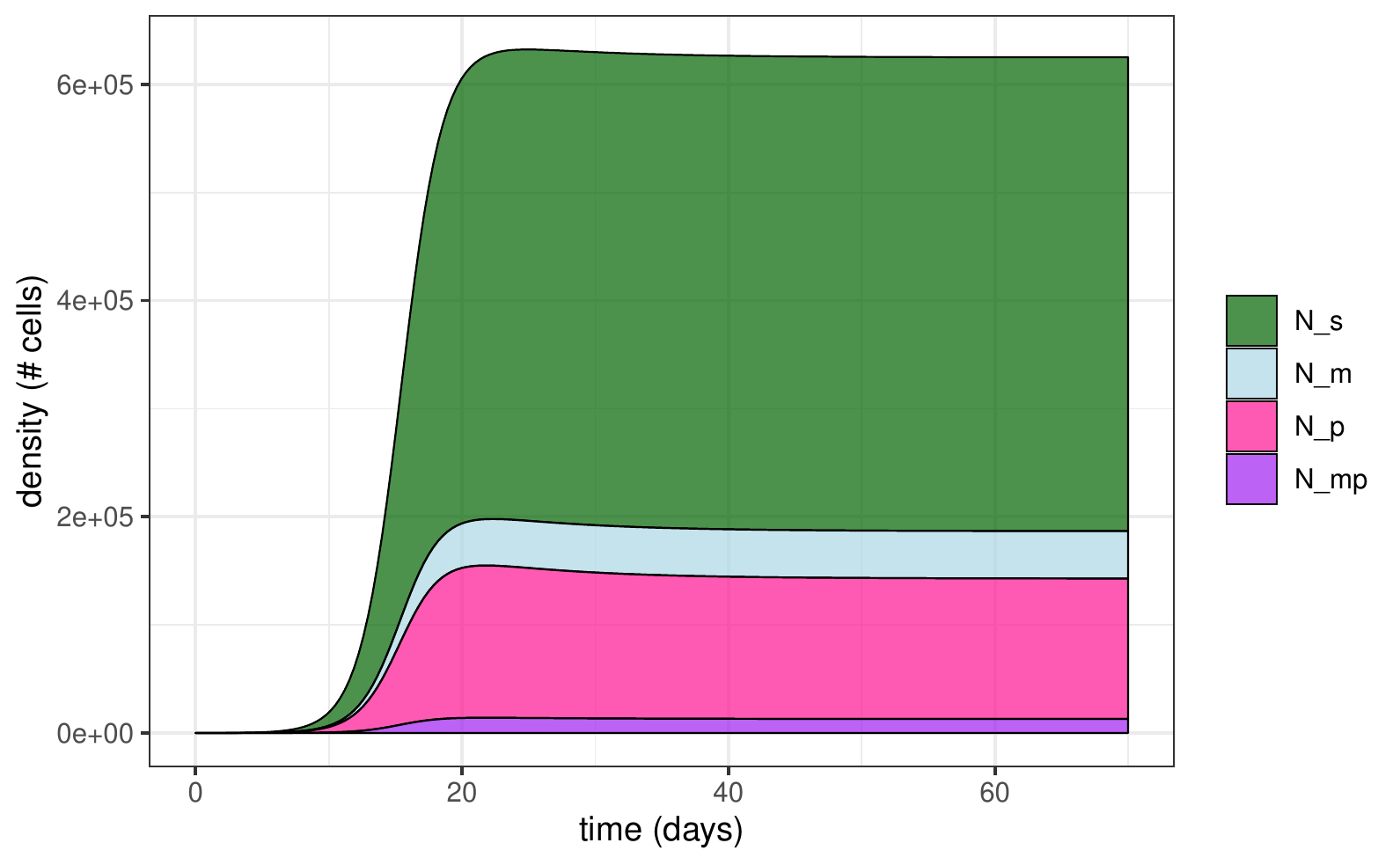}  & \includegraphics[width=.5\textwidth]{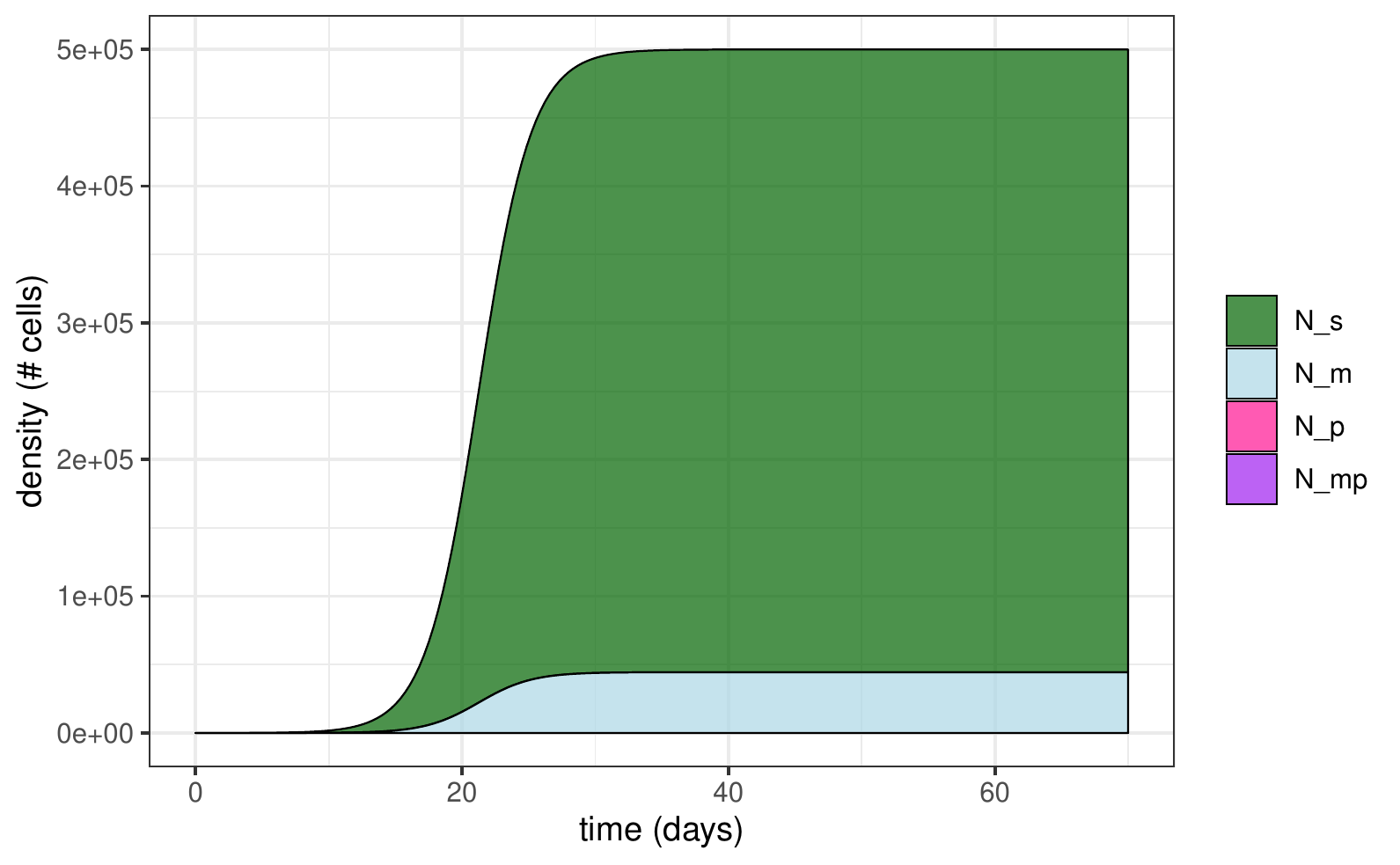} \\
		(C) & (D)
	\end{tabular}
	\caption{Bacterial density chronicles, as decimal logarithm of cell number
		per mL through time in days, obtained through numerical integration
		of System \eqref{eq-GenModel}: hierarchy diversity. Each plot (A) to (D) is representative of the asymptotic
		behavior of the system in the corresponding zone likewise labelled
		in Fig. \ref{DiagStability}. Shared parameters between plots are, $\Lambda=25,d=d_{j}=5,\beta_{j}=10^{-6}$,$\varepsilon_{j}=10^{-8},\theta=0.25,\alpha=10^{-13},a_H=1,b_H=0$,
		with $j\in\mathcal{J}$ (see Table \ref{Tab-ModelParameters} for units). Initial conditions
		are $\left(B,N_{s},N_{m},N_{p},N_{m.p}\right)\left(0\right)=\left(5,10,1,10,1\right)$.
		With $\mathbf{k}\coloneqq10^{-6}\cdot\left(\tau_{s},\tau_{m},\tau_{p},\tau_{m.p}\right)$,
		varying parameters are: (A) $\mathbf{k}=\left(0.9,0.9,0.9,0.9\right)$,
		(B) $\mathbf{k}=\left(0.9,0.9,10,10\right)$, (C) $\mathbf{k}=\left(1.1,1.1,1.5,1.5\right)$,
		(D) $\mathbf{k}=\left(1.1,1.1,0.9,0.9\right)$. }
	\label{fig:my_label}
\end{figure}

\begin{figure}
	\centering
	\begin{tabular}{cc}
		\includegraphics[width=.5\textwidth]{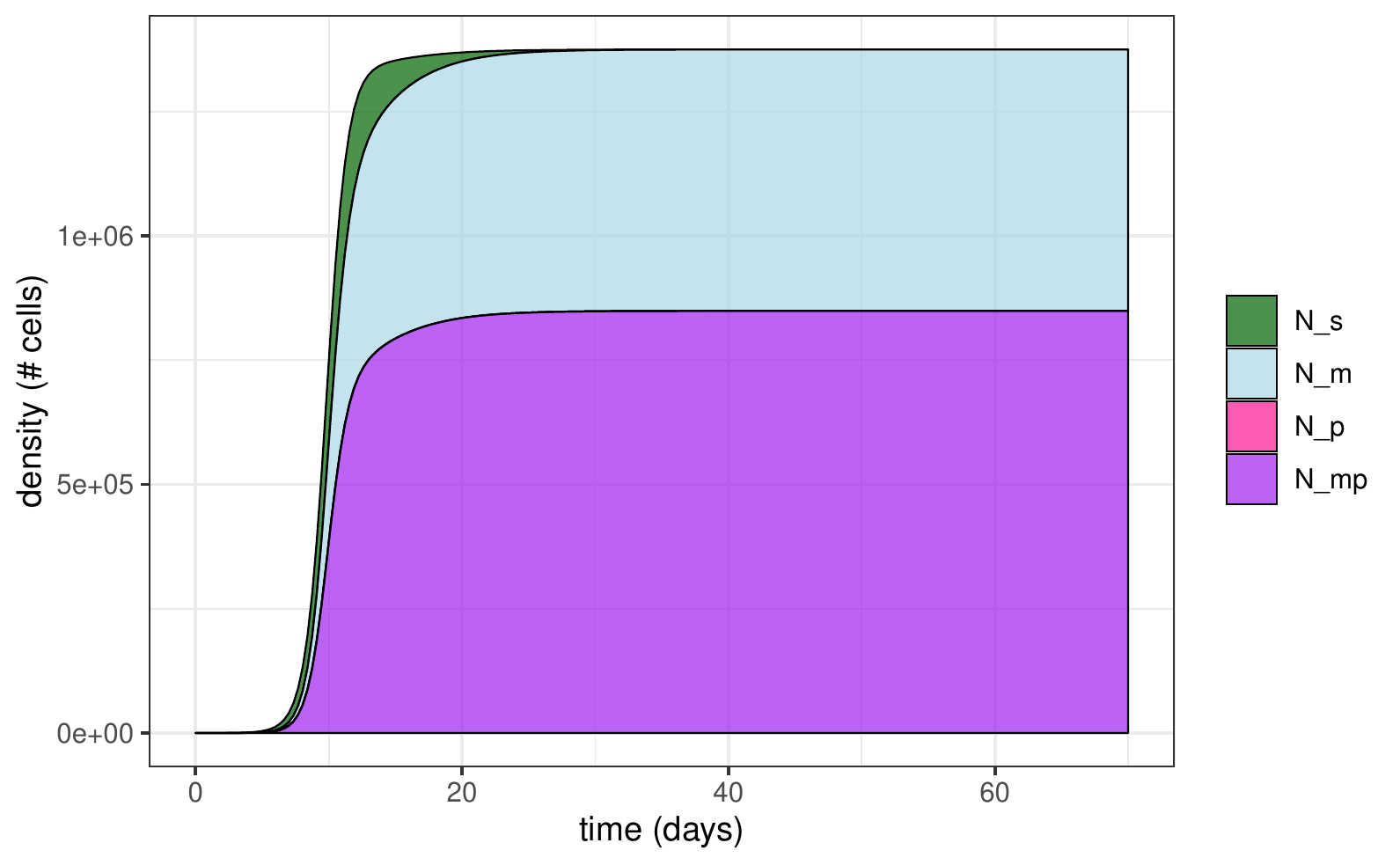}  & \includegraphics[width=.5\textwidth]{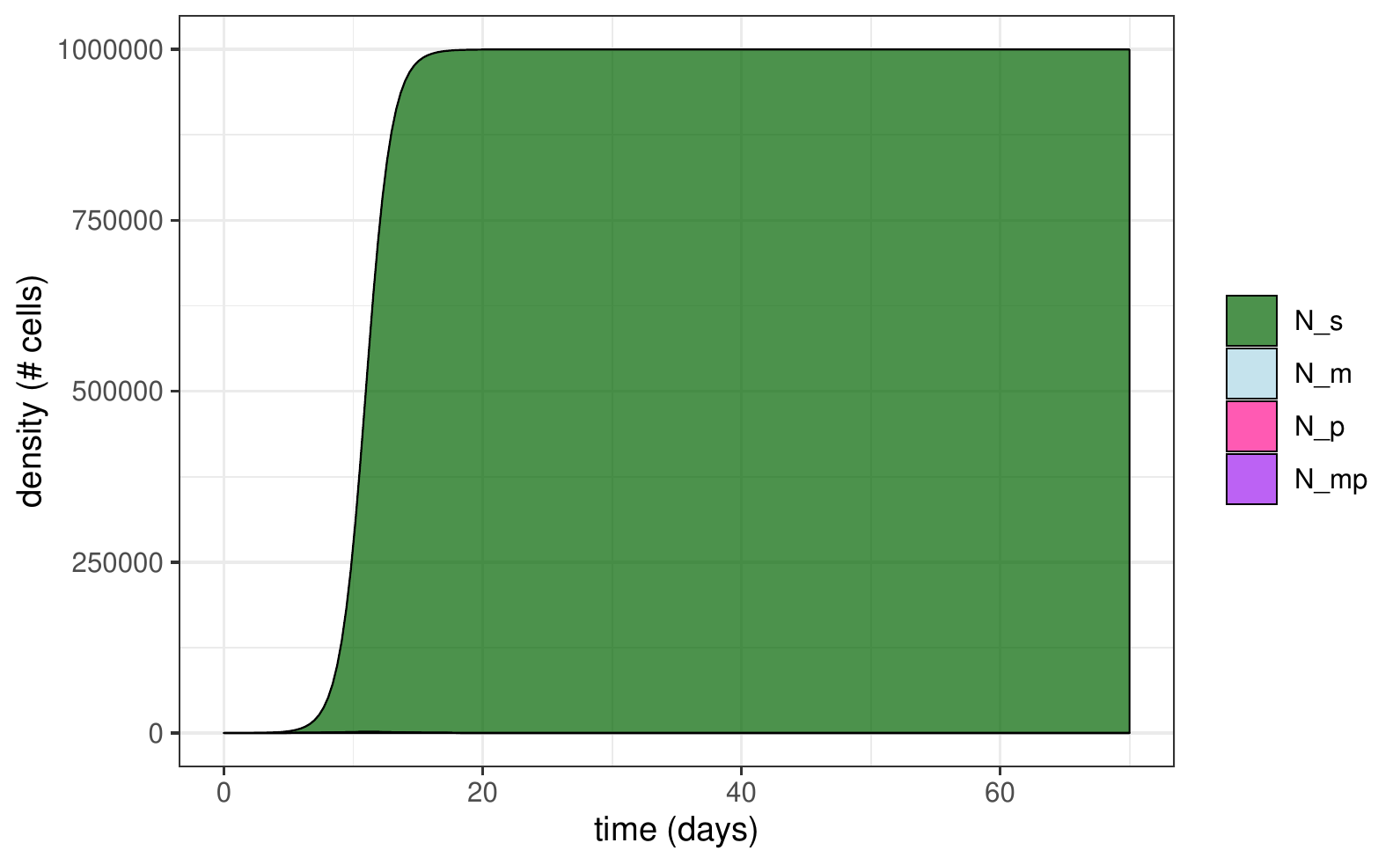} \\
		(C) & (D)
	\end{tabular}
	\caption{Bacterial density chronicles, as decimal logarithm of cell number
		per mL through time in days, obtained through numerical integration
		of system \eqref{eq-GenModel} for treatment-free scenarios: hierarchy diversity. Plots (C) and (D) correspond to zones labelled in Fig. \ref{DiagStability}. All parameters are equal to those of \ref{fig:my_label} except for: (C) $\mathbf{k}=\left(1.2,1.1,1.467,1.7\right)$, (D) $\mathbf{k}=\left(1.2,1.1,1.467,1.4\right)$, where $\mathbf{k}\coloneqq10^{-6}\cdot\left(\tau_{s},\tau_{m},\tau_{p},\tau_{m.p}\right)$.}
	\label{fig:TreatFree}
\end{figure}

\begin{figure}
	\centering
	\begin{tabular}{ccc}
		\includegraphics[width=.4\textwidth]{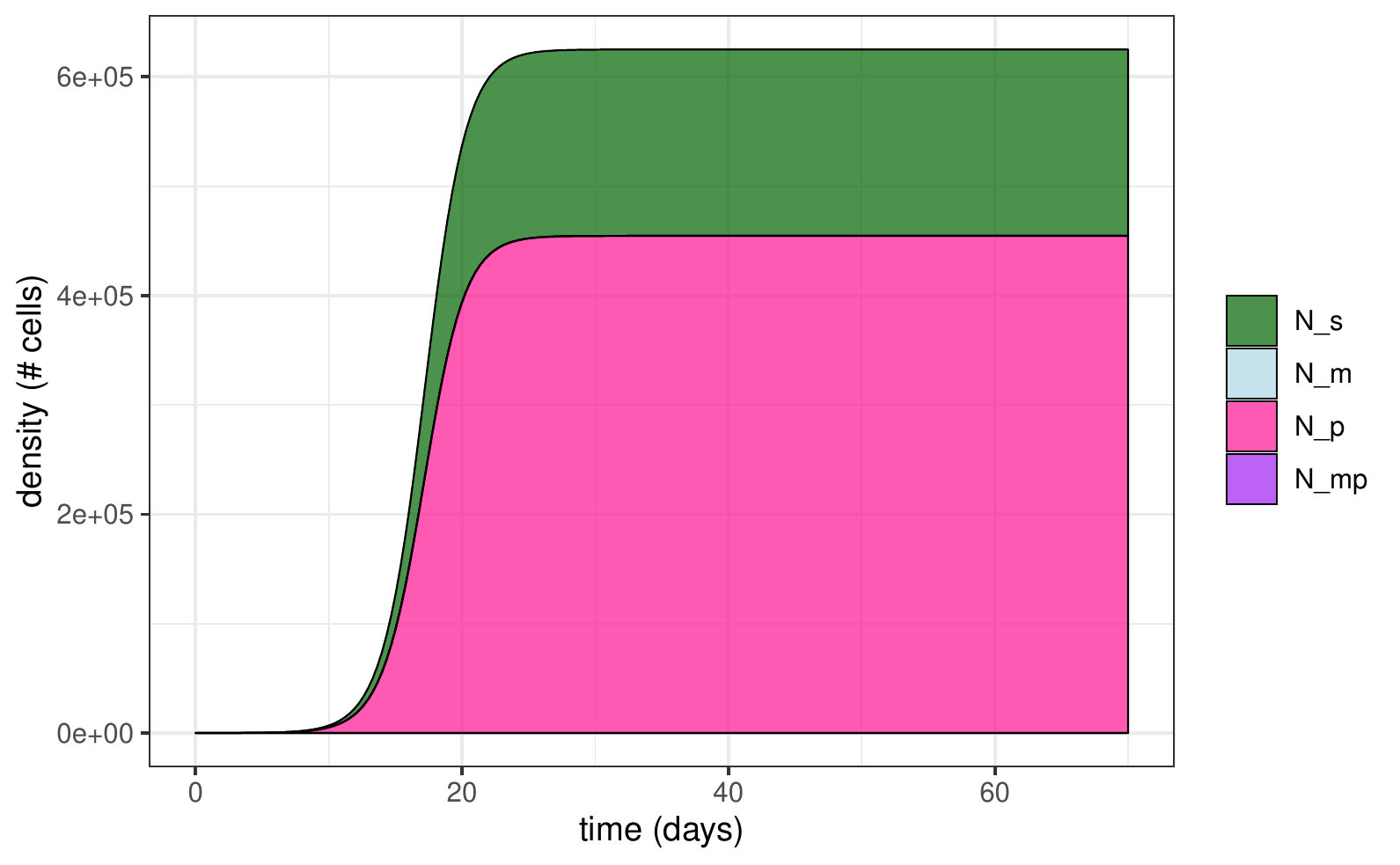}  & \includegraphics[width=.4\textwidth]{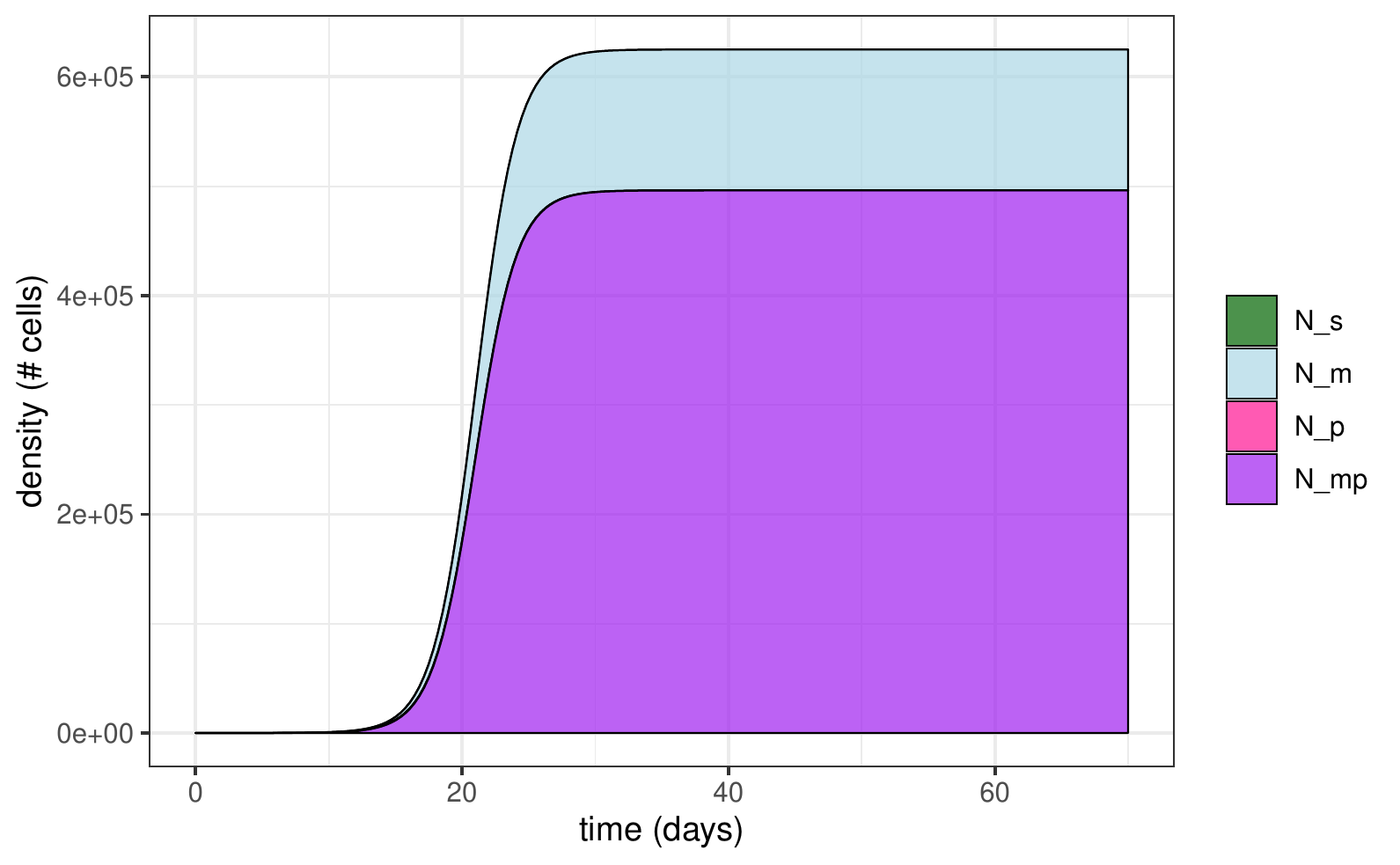}
		& \includegraphics[width=.4\textwidth]{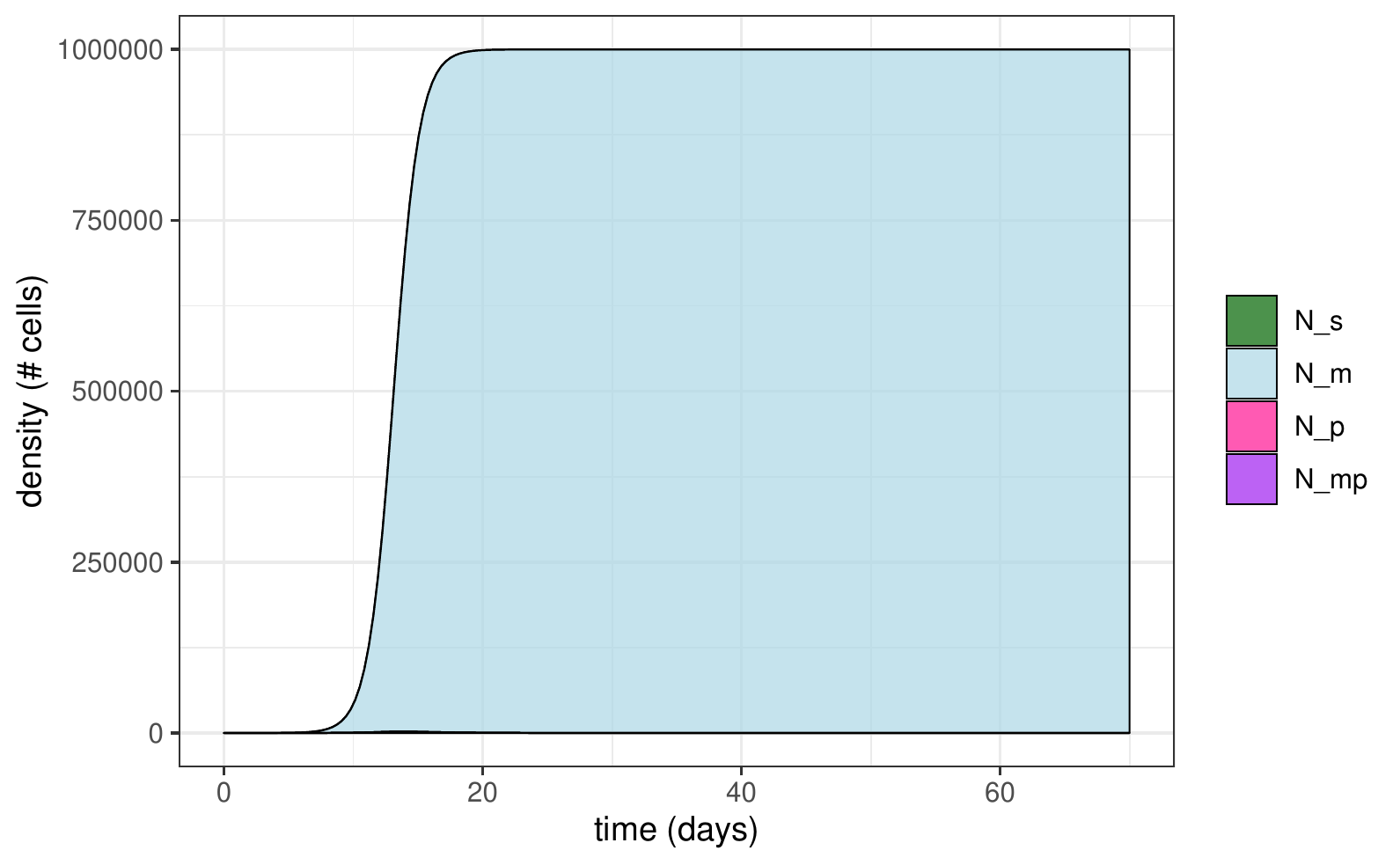}\\
		(B) & (C) & (D) 
	\end{tabular}
	\caption{Bacterial density chronicles, as decimal logarithm of cell number per mL through time in days, obtained through numerical integration of System \eqref{eq-GenModel} for treatment scenarios: hierarchy diversity. Plots (B to D) correspond to zones labelled in Fig. \ref{DiagStability}. All parameters are equal to those of \ref{fig:my_label} except for: (B) $\mathbf{k}=\left(0.9,0.8,1.5,1.4\right)$, (C) $\mathbf{k}=\left(0.9,0.8,1.4,1.5\right)$, (D) $\mathbf{k}=\left(0.9,1.2,1.1,1.5\right)$, where $\mathbf{k}\coloneqq10^{-6}\cdot\left(\tau_{s},\tau_{m},\tau_{p},\tau_{m.p}\right)$.}
	\label{fig:WithTreat}
\end{figure}

Here, we check the local asymptotic stability of stationary states $E^*_{s-m}$ and $E^*$ by linearizing System \eqref{eq-GenModel} at those points. Stability results are summarised in Figures \ref{DiagStability}-\ref{fig:my_label} and given by the following theorem (see Section \ref{proof of thm-stab-reults} for the proof).
\begin{theorem} \label{thm-stab-reults}
	Let mutation ratios and flux rate of HGT  be small enough.
	\begin{description}
		\item[(i)] If $\max \left(\mathcal{R}_p,\mathcal{R}_{m.p} \right) >1 > \max \left(\mathcal{R}_s,\mathcal{R}_m \right)$ then, Model \eqref{eq-GenModel} have two stationary states: $E^0$ (unstable) and $E^*$ (locally asymptotically stable -l.a.s. for short).
		\item[(ii)] If $\max \left(\mathcal{R}_p,\mathcal{R}_{m.p} \right) >\max \left(\mathcal{R}_s,\mathcal{R}_m \right)> 1 $ then, Model \eqref{eq-GenModel} have three stationary states: $E^0$ (unstable), $E^*_{s-m}$ (unstable) and $E^*$ (l.a.s.).
		\item[(iii)] If $\max \left(\mathcal{R}_s,\mathcal{R}_m \right)> 1$ and $\max \left(\mathcal{R}_s,\mathcal{R}_m \right) > \max \left(\mathcal{R}_p,\mathcal{R}_{m.p} \right) $ then, Model \eqref{eq-GenModel} have two stationary states: $E^0$ (unstable), $E^*_{s-m}= (B^*,N^*_s,N^*_m,0,0)$ which is l.a.s. providing that the following condition is satisfied
 \begin{equation}\label{cond_las_Esm}
\left\{	\begin{split}
	& \frac{\mathcal{R}_p}{\max \left(\mathcal{R}_s,\mathcal{R}_m\right)} +\alpha \frac{N^*_s}{d_p(a_H+b_HN^*)} <1,\\
	& \frac{\mathcal{R}_{m.p}}{\max \left(\mathcal{R}_s,\mathcal{R}_m\right)} + \alpha \frac{N^*_s}{d_{m.p}(a_H+b_HN^*)} <1.
	\end{split}
	\right.
		\end{equation}
	\end{description}
\end{theorem}

Following Theorem \ref{thm-stab-reults}, note that for plasmids to persist at equilibrium, it is necessary that $\max \left(\mathcal{R}_p,\mathcal{R}_{m.p} \right) > \max \left(\mathcal{R}_s,\mathcal{R}_m \right)$. The threshold $\mathcal{R}_0^p= \frac{\max\left(\mathcal{R}_p,\mathcal{R}_{m.p} \right)} {\max\left(\mathcal{R}_s,\mathcal{R}_m \right)}$ can therefore be viewed as the total number of plasmid-bearing strains arising from one strain bearing such a plasmid and introduced into a plasmid-free environment. $\mathcal{R}_0^p$ is similar to the basic reproduction number $\mathcal{R}_0$ in epidemiology and serves as a sharp threshold parameter determining whether or not plasmids can persist by   $\mathcal{R}_0^p<1$  or   $\mathcal{R}_0^p>1$. This also highlights how plasmids differ from mutations: the former are transmitted whereas the latter are not.

\section{Applications}

This section describes how our general analysis can be applied in two situations: (i) bacteria strains interacting without any drug pressure, and (ii) bacteria strains interacting with drug action. Due to our model formulation and results, to explore a scenario we only need to analyse the $\mathcal{R}_j$s.

\subsection{The treatment-free model}
Here, we assume that strains $N_j$s interact without any drug action. In such scenario, we further assume that before the treatment, each  $N_j$ can potentially grow $\left(\mathcal{R}_j >1, \text{ for all, } j \in \mathcal{J} \right)$. We further assume that there is a fitness cost of bearing a resistance, either through mutation or plasmid $\left(\min \left\{\mathcal{R}_p; \mathcal{R}_m \right\}\le  \max \left\{\mathcal{R}_p; \mathcal{R}_m \right\}< \mathcal{R}_s\right)$ and that these costs either add up or compensate each other. 

\paragraph{Costs of mutation and plasmid add up.} In such situation, we have the following realistic assumption: $1<\mathcal{R}_{m.p} < \min \left\{\mathcal{R}_p; \mathcal{R}_m \right\}\le  \max \left\{\mathcal{R}_p; \mathcal{R}_m \right\}< \mathcal{R}_s.$  Then, based on Sections \ref{sec-Inv} and \ref{sec-Equi-Gen-Model} and under the assumptions above, the co-existence of $N_s$ and $N_m$, $E^*_{s-m}= \left(B^*,N^*_s,N^*_m,0,0\right)$, is the only stable stationary state when mutation rates and flux of HGT are small enough (Figure \ref{fig:TreatFree}D). We then have $B^* \simeq B_0/\mathcal{R}_s$ and the corresponding proportion of resistance $P_R(E^*_{s-m})= \mathcal{P}^*_m=\mathcal{O}(\eta)$ is estimated by \eqref{prop Esm Rs is max}.
From a biological standpoint, this means that when the cost of mutation and plasmid add up, in the absence of treatment it is expected that the proportion of resistance is very small at equilibrium.

\paragraph{Costs of mutation and plasmid compensate each other.} In this scenario, we can have two cases:  in case 1 the strain $N_{m.p}$ performs less well than $N_s$ $\left(1<\min \left\{\mathcal{R}_p; \mathcal{R}_m \right\}\le  \max \left\{\mathcal{R}_p; \mathcal{R}_m \right\}< \mathcal{R}_{m.p} < \mathcal{R}_s\right)$ , while in case 2 the compensation is such that strain $N_{m.p}$ performs at least as well as $N_s$ $\left(1<\min \left\{\mathcal{R}_p; \mathcal{R}_m \right\}\le  \max \left\{\mathcal{R}_p; \mathcal{R}_m \right\}< \mathcal{R}_s \le \mathcal{R}_{m.p} \right)$. Case 1 is similar to the previous paragraph where costs add up leading to a negligible proportion of resistance at equilibrium, while case 2 corresponds to a situation where plasmidic resistance spreads and persists (despite the costs it imposes to its hosts) in an environment that is not necessarily subject to drug pressure \cite{Hughes1983,Bergstrom2000,Lili2007}. Indeed, for case 2, results of \ref{sec-Inv} and \ref{sec-Equi-Gen-Model} suggest plasmid invasion and co-existence of all the strains, $E^*= \left(B^*,N^*_s,N^*_m,N^*_p,N^*_{m.p}\right)$, is the only stable stationary state when mutation rates and flux of HGT are small enough (Figure \ref{fig:TreatFree}C). The corresponding proportion of resistance $P_R(E^*)\simeq \mathcal{P}^*_m+ \mathcal{P}^*_{m.p},$
wherein proportions $\mathcal{P}^*_m$ and $ \mathcal{P}^*_{m.p}$ estimated by \eqref{eq resistance decomposed} are respectively attributable to mutational and mutational/plasmidic resistance.

\subsection{The model with drug action}
Our model formulation can easily account for the effect of drug action, either cytostatic (acting on the strain death rate $d_j$) or cytotoxic (acting on the strain growth conversion factor $\tau_j$). This makes it possible to discuss some scenarios of Model \ref{eq-GenModel} with the assumption of a constant drug pressure. Here, we briefly discuss two scenarios.

\paragraph{The drug is efficient on sensitive and mutational strains only $\left(\mathcal{R}_m < \mathcal{R}_s< 1\right)$.} In this scenario, evolving drug resistance requires \textit{de novo} genes that can be only acquired via plasmid (not by mutations). We then have two cases: either the mutation is a compensatory mutation for the plasmid presence $\left(\mathcal{R}_m < \mathcal{R}_s< 1< \mathcal{R}_{p} < \mathcal{R}_{m.p}\right)$ or mutation and plasmid costs add up $\left(\mathcal{R}_m < \mathcal{R}_s< 1< \mathcal{R}_{m.p} < \mathcal{R}_p\right)$. For both cases, results of Sections \ref{sec-Inv} and \ref{sec-Equi-Gen-Model} indicate that the co-existence of all strains, $E^*$, is the only stable stationary state. Estimates \eqref{eq resistance decomposed0} and \eqref{eq resistance decomposed} respectively give, for the first case, $P_R(E^*) \simeq \mathcal{P}^*_m+ \mathcal{P}^*_{m.p}$ (Figure \ref{fig:WithTreat}C) and, for the second case, $P_R(E^*)\simeq \mathcal{P}^*_p$ (Figure \ref{fig:WithTreat}B); where proportions $\mathcal{P}^*_m$, $\mathcal{P}^*_p$ and $ \mathcal{P}^*_{m.p}$ are attributable to mutational, plasmidic and mutational/plasmidic resistance.

\paragraph{The drug is efficient on sensitive and plasmid-bearing strain only $\left (\mathcal{R}_p < \mathcal{R}_s< 1\right)$.} Here, the plasmid is useless but it performs less well than the $N_s$ strain $\left( \mathcal{R}_p < \mathcal{R}_s< 1< \mathcal{R}_{m.p} < \mathcal{R}_m\right)$.
Again, following Sections \ref{sec-Inv} and \ref{sec-Equi-Gen-Model}, only the co-existence stationary state $E^*_{s-m}$ of sensitive en mutational resistance strain is plausible at equilibrium when mutation rates and flux of HGT are small enough. Furthermore, based on estimate \eqref{prop Esm Rm is max}, $E_{s-m}^*$ is mostly a totally mutational resistant stationary state, {\it i.e.} $P_R(E^*_{s-m})=\mathcal{P}^*_m \simeq 1$ (Figure \ref{fig:WithTreat}D).


\section{Discussion} \label{sec-discuss}

\paragraph{Mutations are key factors to maintain strain co-existence at equilibrium.}
Compared to similar studies in the literature, here we do not neglect the evolution of resistance via mutation, which we find to substantially impact the qualitative dynamics of the model. The study in \cite{Svara2011} reaches the opposite conclusion that the co-existence of $N_s$ and $N_m$ or the co-existence of all strains are not plausible at equilibrium because they neglect mutations ratios, such that strain transition $ \mathcal{K}_{s\to m}$ or $\mathcal{K}_{p\to m.p}$ is zero. 

\paragraph{Measurement of the cost of bearing a resistance.}
Drug resistance is known to be traded-off against other components of the bacterial life-cycle through competitive fitness. But, most studies investigating costs of
resistance use growth rate as fitness proxy, which incorporates only a single component of bacterial fitness \cite{Melnyk2015,Vogwill2015}. Indeed, this cost can also act on other bacteria quantitative traits such as the cell's death rate. In our analysis, the cost of resistance is quantified by the effective reproduction numbers $\mathcal{R}_j$s defined by \eqref{eq-threshold-Rj}. These $\mathcal{R}_j$s aggregate potential costs on quantitative traits of the bacteria life cycle such as nutrient consumption rate, nutrient conversion factor, bacteria's death rate, bacteria mutation or plasmid segregation probability.

\paragraph{Occurrence of new mutants.} 
In exponentially growing cells, mutations usually occur during replication \cite{Loewe5650}, but some studies indicate that mutations can be  substantially  higher in  non growing  than  growing cultures \cite{Sniegowski2004}. Thus, the occurrence of new mutants depends either on the abundance of the parental cells or  both the abundance and growth rate of the parental cells \cite{Zur2010}. In Model \eqref{eq-GenModel}, we have first assumed that the latter is satisfied. However, with our model formulation, we can easily  switch to the case where mutations depend on the abundance of the parental cells (see \eqref{eq-GenModel2} for the full model equations). Furthermore, results obtained for Model \eqref{eq-GenModel} remain valid for Model \eqref{eq-GenModel2} with the effective reproduction numbers $\mathcal{R}$ and strains transition numbers $\mathcal{K}$ defined as follows
\begin{equation*}
\begin{split}
& \mathcal{R}_s= \mathcal{T}_s \frac{d_s}{d_s+\varepsilon_s}, \quad  \mathcal{R}_m= \mathcal{T}_m \frac{d_m}{d_m+\varepsilon_m},\\
&\mathcal{R}_p= \mathcal{T}_p (1-\theta)\frac{d_p}{d_p+\varepsilon_p}, \quad \mathcal{R}_{m.p}= \mathcal{T}_{m.p} (1-\theta)\frac{d_{m.p}}{d_{m.p}+\varepsilon_{m.p}},
\end{split}
\end{equation*}
and 
\begin{equation*}
\begin{split}
& \mathcal{K}_{s\to m}= \frac{\varepsilon_s}{d_m+\varepsilon_m}, \quad \mathcal{K}_{m\to s}= \frac{\varepsilon_m}{d_s+\varepsilon_s}\\
& \mathcal{K}_{p\to m.p}= \frac{\varepsilon_p}{d_{m.p}+ \varepsilon_{m.p}}, \quad \mathcal{K}_{m.p\to p}= \frac{\varepsilon_{m.p}}{d_p+\varepsilon_p}\\
& \mathcal{K}_{p\to s}= B_0\frac{\theta\tau_p\beta_p}{d_s +\varepsilon_s},\quad  \mathcal{K}_{m.p \to m}= B_0\frac{\theta\tau_{m.p}\beta_{m.p}}{d_m+ \varepsilon_m},
\end{split}
\end{equation*}
and wherein $\mathcal{T}_j$s are the same as for the previous model \eqref{eq-GenModel}. Note that here, terms $\varepsilon_j$s represent mutation rates rather than ratios as compared to Model \eqref{eq-GenModel}.

\paragraph{Competition at the plasmid level of selection.} In Model \eqref{eq-GenModel} we assume that plasmid-bearing strains always carry drug resistance genes. This assumption does not have any impact in the case of the treatment-free model. However, although a significant proportion of plasmid-bearing strains is involved in drug resistance, a small proportion can be `non-resistant plasmids'. To focus even more on the question of drug action with the model, it could be interesting to introduce additional bacteria strains with non-resistant plasmids but paying a cost of plasmid carriage, as in \cite{Svara2011}. This would allow to consider competition between plasmidic strains carrying the drug resistance gene with other plasmidic strains that do not carry the drug resistance gene. Indeed, competition at the plasmid level can be of great importance, since the spread of a resistant plasmid can be slowed or entirely stopped by a nonresistant version of the same plasmid \cite{Tazzyman2014}. This issue will will be rigorously investigated in a forthcoming work.


\begin{appendix}
\setcounter{table}{0}
\let\oldthetable\thetable
\renewcommand{\thetable}{S\oldthetable}

\setcounter{figure}{0}
\let\oldthefigure\thefigure
\renewcommand{\thefigure}{S\oldthefigure}

\setcounter{equation}{0}

\section{Proof of Theorem \ref{thm-exist}}\label{proof of thm-exist}
	The right-hand side of System \eqref{eq-GenModel} is continuous and locally lipschitz on $\R^5$. Using a classic existence theorem, we then find $T>0$ and a unique solution $E(t)\omega_0= \left( B(t),N_s(t),N_m(t),N_p(t),N_{m.p}(t)\right)$ of \eqref{eq-GenModel} from $[0,T) \to \R^5$ and passing through the initial data $\omega_0$ at $t=0$.  Let us now check the positivity and boundedness of the solution $E$ on $[0,T)$. 
	
	Since $E(\cdot)\omega_0$ starts in the positive orthant $\R^5_+$, by continuity, it must cross at least one of the five borders $\{B=0\}$, $\{N_j=0\}$ (with $j\in \mathcal{J}$) to become negative. Without loss of generality, let us assume that $E$ reaches the border $\{N_s=0\}$. This means we can find $t_1 \in (0,t)$ such that $N_j(t)>0$ for all $t \in(0,t_1)$, $j \in \mathcal{J}$ and $N_s(t_1)=0$, $B(t_1) \ge 0$, $N_j(t_1) \ge 0$ for $j \in \mathcal{J}$. 
	Then, the $\dot N_s$-equation of \eqref{eq-GenModel}  yields $\dot N_s(t_1)=  \theta\tau_p\beta_p B(t_1)N_p(t_1) \ge 0$, from where the orbit $E(\cdot)\omega_0$ cannot cross $\R^5_+$ through the border $\{N_s=0\}$. Similarly, we prove that at any borders $\{B=0\}$, $\{N_j=0\}$ (with $j \in \mathcal{J}$), either the resulting vector field stays on the border or  points inside $\R^5_+$. Consequently, $E([0,T))\omega_0 \subset \R^5_+$.
	
	Recalling that $N=\sum_{j \in \mathcal{J}} N_j$ and adding up the $\dot N$- and $\dot B$-equations, it comes 
	\[
	\frac{\d}{\d t} \left(\tau_{\text{max}}B + N \right)\le \tau_{\text{max}}\Lambda - \min \left(d,d_{\text{min}} \right) \left(\tau_{\text{max}}B + N \right);
	\]
	from where one deduces estimate \eqref{eq-bounded}. So, the aforementioned local solution of System \eqref{eq-GenModel} is a global solution {\it i.e.} defined for all $t\in \R_+$. Which ends the proof of Theorem \ref{thm-exist}.

%
%
%
%
%
%
%

\section{Proof of Theorem \ref{Thm-invasion}}\label{proof of Thm-invasion}
If we consider a small perturbation of the bacteria-free steady state $E^0$, the initial phase of the invasion can be described by the linearized system at $E^0$. Since the linearized equations for bacteria populations do not include the one for the nutrient, we then have 
\begin{equation}\label{eq-Linear-E0}
\begin{split}
\frac{\d}{\d t} (u,v)^T= J[E^0] (u,v)^T,
\end{split}
\end{equation}
with $ J[E^0]= \left(\begin{array}{cc}
B_0G-D & B_0L_p\\
0 & B_0(G_p-L_p)-D_p
\end{array} \right).$ 

We claim that 
\begin{claim}\label{claim-spectre-DFE}
For small mutation rates $\varepsilon_v$, the principal eigenvalue $r(B_0G-D )$ and $r( B_0(G_p-L_p)-D_p)$, of matrices $(B_0G-D )$ and $( B_0(G_p-L_p)-D_p)$ writes 
\[
\begin{split}
r(B_0G-D )=& \frac{1}{2} \left\{B_0\tau_s\beta_s(1-\varepsilon_s)-d_s + B_0\tau_m\beta_m(1-\varepsilon_m) -d_m + \right.\\
& \left. \left[ \left(B_0\tau_s\beta_s(1-\varepsilon_s)-d_s - B_0\tau_m\beta_m(1-\varepsilon_m)+d_m \right)^2 +4\varepsilon_m\varepsilon_s\tau_m\tau_s\beta_m\beta_sB_0^2 \right]^{1/2} \right\}
\end{split}
\]
{\it i.e.}
\[
r(B_0G-D )= \left\{
\begin{split}
&B_0\tau_s\beta_s(1-\varepsilon_s)-d_s + \mathcal{O}(\varepsilon_m\varepsilon_s),  \quad \text{ if } \mathcal{R}_s> \mathcal{R}_m,\\
&B_0\tau_m\beta_m(1-\varepsilon_m)-d_m + \mathcal{O}(\varepsilon_m\varepsilon_s),  \quad \text{ if } \mathcal{R}_s< \mathcal{R}_m,
\end{split}
\right.
\]

By the same way, we have 
\[
r(( B_0(G_p-L_p)-D_p))= \left\{
\begin{split}
&B_0\tau_p\beta_p(1-\theta)(1- \varepsilon_p)-d_p + \mathcal{O}(\varepsilon_p\varepsilon_{m.p}),  \quad \text{ if } \mathcal{R}_p> \mathcal{R}_{m.p},\\
&B_0\tau_{m.p}\beta_{m.p}(1-\theta)(1-\varepsilon_{m.p})-d_{m.p} + \mathcal{O}(\varepsilon_p\varepsilon_{m.p}),  \quad \text{ if } \mathcal{R}_p< \mathcal{R}_{m.p},
\end{split}
\right.
\]
\end{claim}

Denoting by $\sigma(J[E^0])$ the spectrum of $J[E^0]$, we recall that the stability modulus of $J[E^0]$ is $s_0(J[E^0])=\{\max \text{Re}(z): z\in \sigma(J[E^0]) \}$ and $J[E^0]$ is said to be locally asymptotically stable (l.a.s.) if $s_0(J[E^0])<0$ \cite{LiWang1998}. Following Claim \ref{claim-spectre-DFE} it comes
\[
\begin{split}
&\sigma(J[E^0])= \sigma(B_0G-D ) \cup \sigma( B_0(G_p-L_p)-D_p)  \simeq \\
&\left\{ B_0\tau_s\beta_s(1-\varepsilon_s)-d_s, B_0\tau_m\beta_m(1-\varepsilon_m)-d_m, B_0\tau_p\beta_p(1-\varepsilon_p)(1-\theta)-d_p,\right.\\ &\left.  B_0\tau_{m.p}\beta_{m.p}(1-\theta)(1-\varepsilon_{m.p})-d_{m.p} \right\},
\end{split}
\]
where the last approximation holds for small mutation rates.

Note that, when mutation rates are small enough, we obtain $s_0(J[E^0])<0$ if and only if $\mathcal{R}^*<1$, {\it i.e.} $E^0$ is l.a.s if $\mathcal{R}^*<1$ and unstable if $\mathcal{R}^*>1$.

We now check the global stability of $E^0$ when $\mathcal{T}^* <1$. The $\dot B$-equation of \eqref{eq-GenModel} gives $ \dot B \le \Lambda- dB$. Further, $B_0$ is a globally attractive stationary state of the upper equation $ \dot w = \Lambda- dw$, {\it i.e.} $w(t)\to B_0$ as $t\to \infty$. Which gives $w(t)\le B_0$ for sufficiently large time $t$, from where $B(t)\le B_0$ for sufficiently large time $t$. Combining this last inequality with the total bacteria dynamics described by \eqref{eq-totalPop}, we find $\dot N\le \sum_{j \in \mathcal{J}}\left( \beta_j\tau_j B_0 -d_j\right)  N\le c_0(\mathcal{T}^*-1)N$, with $c_0>0$ a positive constant. Therefore $N(t) \le N(0) e^{c_0(\mathcal{T}^*-1)t} \to 0$ as as $t\to \infty$. This ends the proof of the global stability of $E^0$ when $\mathcal{T}^* <1$. 

Item (iii) of the theorem remains to be checked. To do so we will apply results in \cite{Hale1989}. Let us first notice that $E^0$ is an unstable stationary
state with respect to the semiflow $E$. To complete the proof, it is sufficient to show that $W^s(\{E^0\})\cap X_0=\emptyset$, where  $W^s(\{E^0\})= \left\{w \in \Omega: \lim_{t\to \infty}  E(t)w= E^0\right\}$ is the stable set of $\{E^0\}$. To prove this assertion, let us argue by contradiction by assuming that there exists $w \in W^s(\{E^0\})\cap X_0$. We set $E(t)w=(B(t),N_s(t),N_m(t),N_p(t),N_{m.p}(t))$.  The $\dot N$-equation defined by \eqref{eq-totalPop} gives 
\[
\begin{split}
\dot N(t) =& \sum_{j \in \mathcal{J}}d_v \left( \mathcal{T}_j \frac{B(t)}{B_0} -1\right)  N_j(t)\\
\ge & d_{\text{min}} \left(\frac{B(t)}{B_0} \min_j \mathcal{T}_j  -1\right) N(t), \text{ for all time } t.
\end{split}
\]
Since $\min_j \mathcal{T}_j>1$ and it is assumed that $B(t) \to B_0$ as $t\to \infty$, we find that the function function $t \mapsto N(t)=N_s(t)+N_m(t)+N_p(t)+N_{m.p}(t)$ is not decreasing for $t$ large enough. Hence there exists $t_0\ge 0$ such that $N(t) \ge N(t_0)$ for all $t\ge t_0$. Since $N(t_0)>0$, this prevents the component $(N_s,N_m,N_p,N_{m.p})$ from converging to $(0,0,0,0)$ as $t \to \infty$. A contradiction with $E(t)\omega \to E^0$. This completes the proof of Theorem \ref{Thm-invasion}.

\section{Proof of Theorem \ref{thm-equi-Model}}\label{proof of thm-equi-Model}
	\paragraph{The stationary state $E^*_{s-m}= \left(B^*,u^*,0 \right)$.}  Here, it is useful to considered the abstract formulation of the model given by \eqref{eq-GenModelCompact}.
	The $\dot u$-equation of \eqref{eq-GenModelCompact} gives $\left[BG-D\right] u=0$ i.e. $D^{-1}Gu= u/B$. Therefore $1/B= r(D^{-1}G)$ and $u=c\phi$, where $\phi>0$ is the eigenvector of $D^{-1}G$ corresponding to $r(D^{-1}G)$ and normalized such that $\|\phi\|_1=1$ and $c>0$ is a positive constant. Notice that $\phi>0$ means all components of the vector $\phi$ are positive. More precisely, we have $D^{-1}G= \left[
	\begin{array}{cc}
	\mathcal{R}_s/B_0 & \mathcal{K}_{m\to s}/B_0\\
	\mathcal{K}_{s\to m}/B_0 & \mathcal{R}_m/B_0
	\end{array}
	\right]$, 
	\begin{equation*}
	\begin{split}
	& B^{-1}=r(D^{-1}G)= \frac{1}{2B_0} \left\{\mathcal{R}_s +\mathcal{R}_m + \left[\left(\mathcal{R}_s -\mathcal{R}_m \right)^2 +4 \mathcal{K}_{m\to s}  \mathcal{K}_{s\to m} \right]^{1/2}   \right\},\\
	& \phi= \frac{\phi_0}{\|\phi_0\|_1},
	\end{split}
	\end{equation*}
	with $\phi_0= \left(\mathcal{R}_s -\mathcal{R}_m + \left[\left(\mathcal{R}_s -\mathcal{R}_m \right)^2 +4 \mathcal{K}_{m\to s}  \mathcal{K}_{s\to m} \right]^{1/2}  ,2 \mathcal{K}_{s\to m}\right)^T$.
	
	Further, from the $\dot B$-equation we find $\Lambda-dB= c  B \left< (\beta_s,\beta_m),\phi\right>$, {\it i.e. }
	\begin{equation*}
	c = \frac{d\left(B_0r(D^{-1}G)-1\right)}{\left< (\beta_s,\beta_m),\phi\right>} >0 \Longleftrightarrow B_0r(D^{-1}G)>1.
	\end{equation*}

	\paragraph{Approximation of $P_R \left(E^*_{s-m}\right)$ for small mutation rates.}
	Now, let us assumed that mutation rates $\varepsilon_v$ are small enough. Without loss of generality, we express parameters $\epsilon_j$ as functions of the same quantity, let us say $\eta$, with $\eta \ll 1$. We have 
	\begin{equation}\label{eq-Equi-Esm}
	   P_R \left(E^*_{s-m}\right)= \frac{2 \mathcal{K}_{s\to m}} {\mathcal{R}_s -\mathcal{R}_m + \left[\left(\mathcal{R}_s -\mathcal{R}_m \right)^2 +4 \mathcal{K}_{m\to s}  \mathcal{K}_{s\to m} \right]^{1/2}  +2 \mathcal{K}_{s\to m}}. 
	\end{equation}
	By setting $\zeta_\eta =\mathcal{R}_s -\mathcal{R}_m + \left[\left(\mathcal{R}_s -\mathcal{R}_m \right)^2 +4 \mathcal{K}_{m\to s}  \mathcal{K}_{s\to m} \right]^{1/2}$, it comes 
	\begin{equation}\label{eq-Equi-cEsm}
	    \zeta_\eta= \left( \mathcal{T}_s -\mathcal{T}_m\right) (1-\eta)+ |\mathcal{T}_s -\mathcal{T}_m| \left[1- \eta + \frac{1}{2} \left( \left( \frac{\mathcal{T}_s +\mathcal{T}_m}{\mathcal{T}_s -\mathcal{T}_m} \right)^2  -1\right)\eta^2 \right] +\mathcal{O}(\eta^3).
	\end{equation}
	
	Next, we find a simple approximation of the frequency  $P_R \left(E^*_{s-m}\right)$ for cases $\mathcal{R}_s>\mathcal{R}_m$ and $\mathcal{R}_s<\mathcal{R}_m$, {\it i.e. } $\mathcal{T}_s>\mathcal{T}_m$ and $\mathcal{T}_s<\mathcal{T}_m$ for small mutations $\varepsilon_j$'s.

	{\it Case: $\mathcal{T}_s>\mathcal{T}_m$.} From estimates \eqref{eq-Equi-Esm} and \eqref{eq-Equi-cEsm} it comes 
	\[ 
	\left\{
	\begin{split}
	   & P_R \left(E^*_{s-m}\right)= \frac{d_s}{d_m} \frac{\mathcal{T}_s}{\mathcal{T}_s- \mathcal{T}_m} \eta + \mathcal{O}(\eta^2),\\
	   &\text{with} \quad \mathcal{R}_s= \mathcal{T}_s (1-\varepsilon_s)  >1, \quad \text{and   } \mathcal{T}_s > \mathcal{T}_m.
	\end{split}
	\right.
	\]

	{\it Case: $\mathcal{T}_s<\mathcal{T}_m$.} From estimates \eqref{eq-Equi-Esm} and \eqref{eq-Equi-cEsm} it comes 
	\[ 
	\left\{
	\begin{split}
	   & P_R \left(E^*_{s-m}\right)= 1- \frac{d_m \left(\mathcal{T}_m-\mathcal{T}_s \right) }{4B_0\tau_s\beta_s} \left( \left( \frac{\mathcal{T}_s +\mathcal{T}_m}{\mathcal{T}_s -\mathcal{T}_m} \right)^2  -1\right) \eta  + \mathcal{O}(\eta^2),\\
	   &\text{with} \quad \mathcal{R}_m= \mathcal{T}_m (1-\varepsilon_m)  >1, \quad \text{and   } \mathcal{T}_m > \mathcal{T}_s.
	\end{split}
	\right.
	\]

	\paragraph{The stationary state $E^*= \left(B^*,N_s^*,N_m^*,N_p^*,N^*_{m.p} \right)$.} By setting $u=(N_s,N_m)^T$, $v=(N_p,N_{m.p})^T$ 
	we have from $\dot u= \dot v=0$
	
	\[
	\left\{
	\begin{split}
	& B\left(G u+ L_pv\right)=Du+ H(N)[N_p+N_{m.p}] \left(N_s, N_m \right)^T,\\
	& B(G_p-L_p)v= D_pv-  H(N)[N_p+N_{m.p}] \left(N_s, N_m \right)^T ,
	\end{split}
	\right.
	\]
	i.e.
	\begin{equation}\label{eq-for-psi-fullEqui}
	B L (u,v)+\alpha F(u,v) = (u,v),
	\end{equation}
	with 
	\[
	\begin{split}
	&L= \mathcal{D}^{-1} \mathcal{G}= \left[
	\begin{array}{cc}
	D^{-1}G & D^{-1}L_p\\
	0 & D_p^{-1}(G_p-L_p)
	\end{array}
	\right], \\
	& \mathcal{G}= 
	\left[
	\begin{array}{cc}
	G & L_p\\
	0 & G_p-L_p
	\end{array}
	\right]; \quad \mathcal{D}= 
	\left[
	\begin{array}{cc}
	D & 0\\
	0 & D_p
	\end{array}
	\right],\\
	& F(u,v)= \frac{N_p+ N_{m.p}}{a_H+b_HN} \left[ 
	\begin{array}{c}
	-D^{-1}u  \\
	D_p^{-1}u  
	\end{array}
	\right]. 
	\end{split}
	\]
	Note that the spectrum of $L$ is $\sigma(L)= \sigma\left(D^{-1}G\right)\cup \sigma \left(D_p^{-1}(G_p-L_p)\right)$ and the spectral radius of matrices $D^{-1}G$ and $D_p^{-1}(G_p-L_p)$ are given by
	\[
	\begin{split}
	& r(D^{-1}G)= \frac{1}{2B_0} \left\{\mathcal{R}_s +\mathcal{R}_m + \left[\left(\mathcal{R}_s -\mathcal{R}_m \right)^2 +4 \mathcal{K}_{m\to s}  \mathcal{K}_{s\to m} \right]^{1/2}   \right\},\\
	& r(D_p^{-1}(G_p-L_p))= \frac{1}{2B_0} \left\{\mathcal{R}_p +\mathcal{R}_{m.p} + \left[\left(\mathcal{R}_p -\mathcal{R}_{m.p} \right)^2 +4 \mathcal{K}_{m.p\to p}  \mathcal{K}_{p\to m.p} \right]^{1/2}   \right\}.
	\end{split}
	\]

	We now introduced a parametric representation of the stationary state $E^*$ with respect to the small parameter $\alpha$. Using the Lyapunov-Schmidt expansion (see \cite{Cushing1998} for more details), the expanded variables are
	\begin{equation}\label{expand1}
	\begin{split}
	&u= u^0+ \alpha u^1+ \cdots,\\
	&v= v^0+ \alpha v^1+ \cdots,\\
	&B=  b^0+ \alpha b^1+ \cdots,\\
	& F(u,v)= \frac{N_p^0+ N_{m.p}^0}{a_H+b_HN^0} \left[ 
	\begin{array}{c}
	-D^{-1}u^0  \\
	D_p^{-1}u^0 
	\end{array}
	\right] + \cdots,
	\end{split}
	\end{equation}
	with  $u^0=(N_s^0,N_m^0)$, $v^0=(N_p^0,N_{m.p}^0)$ and $N^0=N_s^0+N_m^0+N_p^0+N_{m.p}^0.$
	
	Evaluating the substitution of expansions \eqref{expand1} into the
	eigenvalue equation \eqref{eq-for-psi-fullEqui} at $\mathcal{O}(\alpha^0)$ produces $b_0 L (u^0,v^0)= (u^0,v^0)$, {\it i.e.}
	\begin{equation}\label{system-u0v0-fullEqui}
	\left\{\begin{split}
	& b^0D^{-1}Gu^0 + b^0D^{-1}L_pv^0=u^0,\\
	& b^0D_p^{-1}(G_p-L_p)v^0=v^0.
	\end{split}\right.
	\end{equation}
	Since we are interesting for $v_0>0$, System \eqref{system-u0v0-fullEqui} leads to $b^0 D_p^{-1}(G_p-L_p)v^0=v^0$. The irreducibility of the matrix $D_p^{-1}(G_p-L_p)$ gives that $(1/b_0,v^0)$ is the principal eigenpair of $D_p^{-1}(G_p-L_p)$. That is $1/b^0= r(D_p^{-1}(G_p-L_p))$ and $v^0=  c_0 \frac{\varphi_0}{\|\varphi_0\|_1} $, wherein $c_0$ is a positive constant and\\ $\varphi_0= \left(\mathcal{R}_p -\mathcal{R}_{m.p} + \left[\left(\mathcal{R}_p -\mathcal{R}_{m.p} \right)^2 +4 \mathcal{K}_{m.p\to p}  \mathcal{K}_{p\to m.p} \right]^{1/2}  ,2 \mathcal{K}_{p\to m.p}\right)^T.$
	
	Again, System \eqref{system-u0v0-fullEqui} gives 
	\begin{equation}\label{eq-u0-fullEqui}
	(D^{-1}G -1/b^0\I)u^0 = -D^{-1}L_pv^0. 
	\end{equation}
	Since $D^{-1}L_pv^0 = \text{diag} \left(\mathcal{K}_{p\to s}, \mathcal{K}_{m.p\to m}\right)v^0$ which is positive, equation \eqref{eq-u0-fullEqui} can be solved for $u_0>0$ iff $1/b^0> r\left(D^{-1}G\right)$ and so $u^0 = -(D^{-1}G -1/b^0\I)^{-1} D^{-1}L_pv^0.$ 
	
	From the $\dot B$-equation of the model, the term of order $\mathcal{O}(\alpha^0)$ leads to
	\[
	c_0= \frac{\|\varphi_0\|_1  d\left[B_0/b^0-1 \right] } { \left< \beta, \left(-(D^{-1}G -1/b^0\I)^{-1} D^{-1}L_p\varphi_0, \varphi_0\right) \right> }>0 \Longleftrightarrow  B_0/b^0=B_0r(D_p^{-1}(G_p-L_p)) >1.
	\]

	Consequently, it comes $b_0>0$ and $(u^0,v^0)>0$ are such that 
	\[
	\left\{\begin{split}
	& 1/b^0=r(D_p^{-1}(G_p-L_p)),\\
	& v^0= \frac{  d\left[B_0/b^0-1 \right] } { \left< \beta, \left(-(D^{-1}G -1/b^0\I)^{-1} D^{-1}L_p\varphi_0, \varphi_0\right) \right> } \varphi_0,\\
	& u^0 = -(D^{-1}G -1/b^0\I)^{-1} D^{-1}L_pv^0,
	\end{split}\right.
	\]
	conditioned by 
	\begin{equation}\label{condi-full-Equi}
	r(D_p^{-1}(G_p-L_p))> r\left(D^{-1}G\right) \quad \text{ and } \quad B_0r(D_p^{-1}(G_p-L_p)) >1.
	\end{equation}
	
	Again, evaluating the substitution of expansions \eqref{expand1} into the
	eigenvalue equation \eqref{eq-for-psi-fullEqui} at $\mathcal{O}(\alpha)$ produces 
	\begin{equation*}
	\left(b^0L-\I\right)(u^1,v^1) =- \frac{b^1}{b^0} (u^0,v^0) - \frac{N_p^0+ N_{m.p}^0}{a_H+b_HN^0} \left[ 
	\begin{array}{c}
	-D^{-1}u^0  \\
	D_p^{-1}u^0 
	\end{array}
	\right].
	\end{equation*}
	that is
	\begin{subequations}\label{eq-order-psi1} 
		\begin{empheq}[left=\empheqlbrace]{align}
		& b^0D^{-1}Gu^1 + b^0D^{-1}L_pv^1-u^1= - \frac{b^1}{b^0} u^0 + \frac{N_p^0+ N_{m.p}^0}{a_H+b_HN^0} D^{-1}u^0, \label{eq-order-psi1-u1}\\ 
		& b^0D_p^{-1}(G_p-L_p)v^1-v^1= - \frac{b^1}{b^0} v^0 - \frac{N_p^0+ N_{m.p}^0}{a_H+b_HN^0} D_p^{-1}u^0. \label{eq-order-psi1-v1}
		\end{empheq}
	\end{subequations}

	As $1/b^0$ is a characteristic value of $D_p^{-1}(G_p-L_p)$, $(b^0D_p^{-1}(G_p-L_p)-\I )$ is a singular matrix. Thus, for \eqref{eq-order-psi1-v1} to have a solution, the right-hand side of  \eqref{eq-order-psi1-v1} must be orthogonal to the null space of the adjoint $(b^0D_p^{-1}(G_p-L_p)-\I )^{T}$ of $(b^0D_p^{-1}(G_p-L_p)-\I )$. The null space of $(b^0D_p^{-1}(G_p-L_p)-\I )^{T}$ is spanned by $\omega_0$, where $\omega_0^T$ is the eigenvector of $(D_p^{-1}(G_p-L_p))^T$ corresponding to the eigenvalue  $1/b^0$ and normalized such that $\|\omega_0\|_1=1$. The Fredholm condition for the solvability of \eqref{eq-order-psi1-v1} is $  \omega_0 \cdot \left(\frac{b_1}{b_0} v^0 + \frac{N_p^0+ N_{m.p}^0}{a_H+b_HN^0} D_p^{-1}u^0 \right)= 0$. Which gives 
	\[
	b^1= -b^0 \frac{N_p^0+ N_{m.p}^0}{a_H+b_HN^0} \frac{ \omega_0\cdot D_p^{-1}u^0 }{\omega_0 \cdot v^0}.
	\]
	
	\[
	\left\{\begin{split}
	&(N_s,N_m)= u^0+ \alpha u^1+ \mathcal{O}(\alpha^2),\\
	&(N_p,N_{m.p})= v^0+ \alpha v^1+ \mathcal{O}(\alpha^2),\\
	&B=  b^0+ \alpha b^1+ \mathcal{O}(\alpha^2).
	\end{split}\right.
	\]
	
	\paragraph{Approximation of $E^*$ for small mutation rates.} Here we derive a simple approximation of the stationary state $E^*$ when mutation rates $\varepsilon_v$ are small.  Without loss of generality, we express parameters $\epsilon_j$ as functions of the same quantity, let us say $\eta$, with $\eta \ll 1$. First, we have
	\[
	\begin{split}
	   1/b^0= &r(D_p^{-1}(G_p-L_p))= B_0^{-1} \max \left( \mathcal{R}_p, \mathcal{R}_{m.p} \right)+ \mathcal{O}(\eta^2)\\
	   &= (1-\theta)(1-\eta)  B_0^{-1} \max \left( \mathcal{T}_p, \mathcal{T}_{m.p} \right)+ \mathcal{O}(\eta^2).
	\end{split}
	\]

	Next, we consider two cases $\mathcal{R}_p> \mathcal{R}_{m.p}$ and $\mathcal{R}_p< \mathcal{R}_{m.p}$, {\it i.e. } $\mathcal{T}_p> \mathcal{T}_{m.p} $ and $\mathcal{T}_p< \mathcal{T}_{m.p}$ for small mutations $\varepsilon_j$'s.
	
	{\it Case: $\mathcal{T}_p> \mathcal{T}_{m.p}$.} Recall that condition \eqref{condi-full-Equi} for the existence of the stationary state $E^*$ simply rewrites 
	\begin{equation*}
	\mathcal{R}_p> \max \left( \mathcal{R}_s, \mathcal{R}_m \right)\quad \text{ and } \quad \mathcal{R}_p >1,
	\end{equation*}
	which, for small $\eta$, rewrites
	\begin{equation*}
	(1-\theta)\mathcal{T}_p> \max \left( \mathcal{T}_s, \mathcal{T}_m \right)\quad \text{ and } \quad (1-\theta)\mathcal{T}_p >1.
	\end{equation*}
	We have $\varphi_0= \left((\mathcal{T}_p -\mathcal{T}_{m.p})(1-\theta)(1-\eta) , \eta B_0\tau_p\beta_p/d_{m.p}\right)^T + \mathcal{O}(\eta^2)$ and $\omega_0= \left( (\mathcal{T}_p -\mathcal{T}_{m.p})(1-\theta)(1-\eta) , \eta B_0\tau_{m.p}\beta_{m.p}/d_p \right) + \mathcal{O}(\eta^2)$. Which gives the following approxiamtion of the stationary state $E^*$:
	\begin{equation}\label{esti-Estar-Rp}
	\begin{split}
	& B^*=  \frac{B_0 }{(1-\theta)\mathcal{T}_p} \left(1+ \eta \right) + \mathcal{O}(\eta^2),\\
	& (N_p^*,N_{m.p}^*)=   \frac{c_0}{\Delta^\eta_{p-m.p} + \eta \frac{B_0\tau_p\beta_p}{d_{m.p}}} \left(\Delta^\eta_{p-m.p} , \eta \frac{B_0\tau_p\beta_p}{d_{m.p}} \right)^T   + \mathcal{O}(\eta^2),\\
	&(N_s^*,N_m^*)=  \frac{B_0c_0 \theta} {\Delta^\eta_{p-s}  \left(\Delta^\eta_{p-m.p} + \eta \frac{B_0\tau_p\beta_p}{d_{m.p}}\right)} \times\\
	&   \left( \Delta^\eta_{p-m.p}  \frac{B_0\tau_p\beta_p}{d_s}, \eta \frac{\Delta^\eta_{p-m.p}  \frac{B_0\tau_s\beta_s}{d_m}   \frac{B_0\tau_p\beta_p}{d_s} + \Delta^\eta_{p-s}  \frac{B_0\tau_{m.p}\beta_{m.p}}{d_m}  \frac{B_0\tau_p\beta_p}{d_{m.p}} } { \Delta^\eta_{p-m} } \right)^T+ \mathcal{O}(\eta^2) 
	\end{split}
	\end{equation}
	with $\Delta^\eta_{p-m.p}= (\mathcal{T}_p -\mathcal{T}_{m.p})(1-\theta)(1-\eta)$, $\Delta^\eta_{p-m}= \mathcal{T}_p (1-\theta)(1-\eta)- \mathcal{T}_m (1-\eta)$, $\Delta^\eta_{p-s}= \mathcal{T}_p (1-\theta)(1-\eta)- \mathcal{T}_s (1-\eta)$, $c_0= \frac{d \left(\mathcal{R}_p -\mathcal{R}_{m.p} + \mathcal{K}_{p\to m.p}\right) \left(\mathcal{R}_p-1 \right) } { \beta^* }>0$ and  
	\[
	\begin{split}
	\beta^*= &\beta_p (\mathcal{R}_p-\mathcal{R}_{m.p}) + \beta_{m.p} \mathcal{K}_{p\to m.p}+ \frac{B_0}{ (\mathcal{R}_s-\mathcal{R}_p) (\mathcal{R}_m-\mathcal{R}_p)} \left\{ \beta_s (\mathcal{R}_p-\mathcal{R}_m)(\mathcal{R}_p-\mathcal{R}_{m.p})  \mathcal{K}_{p\to s}\right.\\
	&\left. +\beta_m \left[(\mathcal{R}_p-\mathcal{R}_{m.p})\mathcal{K}_{s\to m} \mathcal{K}_{p\to s} + (\mathcal{R}_p-\mathcal{R}_s)\mathcal{K}_{m.p\to m} \mathcal{K}_{p\to m.p}\right] \right\}.
	\end{split}
	\]
	
	From where, $P_R(E^*)= \mathcal{P}^*_m+ \mathcal{P}^*_p+ \mathcal{P}^*_{m.p}$, with $\mathcal{P}^*_m= \eta \frac{B_0 \theta} {\Delta^\eta_{p-s}  } 
	   \frac{\Delta^\eta_{p-m.p}  \frac{B_0\tau_s\beta_s}{d_m}   \frac{B_0\tau_p\beta_p}{d_s} + \Delta^\eta_{p-s}  \frac{B_0\tau_{m.p}\beta_{m.p}}{d_m}  \frac{B_0\tau_p\beta_p}{d_{m.p}} } { \Delta^\eta_{p-m} N^\eta }$, $\mathcal{P}^*_p= \frac{\Delta^\eta_{p-m.p}}{N^\eta}$, $\mathcal{P}^*_{m.p}= \eta \frac{B_0\tau_p\beta_p}{d_{m.p} N^\eta}$ and\\ 
	   $N^\eta=\Delta^\eta_{p-m.p} + \eta \frac{B_0\tau_p\beta_p}{d_{m.p}} +  \frac{B_0 \theta} {\Delta^\eta_{p-s}   } 
	 \left( \Delta^\eta_{p-m.p}  \frac{B_0\tau_p\beta_p}{d_s}+ \eta \frac{\Delta^\eta_{p-m.p}  \frac{B_0\tau_s\beta_s}{d_m}   \frac{B_0\tau_p\beta_p}{d_s} + \Delta^\eta_{p-s}  \frac{B_0\tau_{m.p}\beta_{m.p}}{d_m}  \frac{B_0\tau_p\beta_p}{d_{m.p}} } { \Delta^\eta_{p-m} } \right).$ 
	 
	 With the Taylor expansion it comes 
	 \[
	 \begin{split}
	 \mathcal{P}^*_m= & \eta \frac{B_0 \theta} {\mathcal{T}_p (1-\theta)- \mathcal{T}_s} 
	   \frac{ (\mathcal{T}_p -\mathcal{T}_{m.p})(1-\theta) \frac{B_0\tau_s\beta_s}{d_m}   \frac{B_0\tau_p\beta_p}{d_s} + \left(\mathcal{T}_p (1-\theta)- \mathcal{T}_s\right)  \frac{B_0\tau_{m.p}\beta_{m.p}}{d_m}  \frac{B_0\tau_p\beta_p}{d_{m.p}} } {  \left(\mathcal{T}_p (1-\theta)- \mathcal{T}_m \right) (\mathcal{T}_p -\mathcal{T}_{m.p})(1-\theta) \left( 1 + \frac{B_0 \theta} {\mathcal{T}_p (1-\theta)- \mathcal{T}_s}   \frac{B_0\tau_p\beta_p}{d_s} \right) }+ \mathcal{O}(\eta^2), \\
	   \mathcal{P}^*_p= & \frac{1}{\left( 1 + \frac{B_0 \theta} {\mathcal{T}_p (1-\theta)- \mathcal{T}_s}   \frac{B_0\tau_p\beta_p}{d_s} \right)} + \mathcal{O}(\eta),\\
	   \mathcal{P}^*_{m.p}=& \eta \frac{B_0\tau_p\beta_p}{d_{m.p} (\mathcal{T}_p -\mathcal{T}_{m.p})(1-\theta) \left( 1 + \frac{B_0 \theta} {\mathcal{T}_p (1-\theta)- \mathcal{T}_s}   \frac{B_0\tau_p\beta_p}{d_s} \right)}+ \mathcal{O}(\eta^2).
	 \end{split}
	 \]

	{\it Case: $\mathcal{T}_p< \mathcal{T}_{m.p}$.} Again, condition \eqref{condi-full-Equi} for the existence of the stationary state $E^*$ becomes 
	\begin{equation*}
	\mathcal{R}_{m.p}> \max \left( \mathcal{R}_s, \mathcal{R}_m \right)\quad \text{ and } \quad \mathcal{R}_{m.p} >1,
	\end{equation*}
	which, for sufficiently $\eta$, rewrites
	\begin{equation*}
	(1-\theta)\mathcal{T}_{m.p}> \max \left( \mathcal{T}_s, \mathcal{T}_m \right)\quad \text{ and } \quad (1-\theta)\mathcal{T}_{m.p} >1.
	\end{equation*}
	
	Again, by Taylor expansion, we find 
	\begin{equation}\label{esti-Estar-Rmp}
	\begin{split}
	   & B^*=  \frac{B_0 }{(1-\theta)\mathcal{T}_{m.p}} \left(1+ \eta \right) + \mathcal{O}(\eta^2),\\
	   &\text{and}\\
	   & v^0=  c_0 \left(0,  1\right) + \mathcal{O}(\eta^2).
	\end{split}
	\end{equation}
Since $u^0=  -(D^{-1}G -1/b^0\I)^{-1} D^{-1}L_pv^0,$ we then have 
\[
	u^0=  \frac{c_0 \theta B_0^2 \tau_p\beta_p }{d_m \left[\mathcal{T}_{m.p}(1-\theta)(1-\eta)-\mathcal{T}_m(1-\eta) \right]  } \left(0,1\right) + \mathcal{O}(\eta^2).
	\]
	From where, $P_R(E^*)= \mathcal{P}^*_m+ \mathcal{P}^*_p+ \mathcal{P}^*_{m.p}$, with
	\[
	\begin{split}
	\mathcal{P}^*_m=& \frac{\theta B_0^2 \tau_p\beta_p} {\theta B_0^2 \tau_p\beta_p +d_m \left[\mathcal{T}_{m.p}(1-\theta)-\mathcal{T}_m \right] }  + \mathcal{O}(\eta),\\
	\mathcal{P}^*_p= &  \mathcal{O}(\eta),\\
	 \mathcal{P}^*_{m.p}= & \frac{d_m \left[\mathcal{T}_{m.p}(1-\theta)-\mathcal{T}_m \right]} {\theta B_0^2 \tau_p\beta_p +d_m \left[\mathcal{T}_{m.p}(1-\theta)-\mathcal{T}_m \right] } + \mathcal{O}(\eta).
	\end{split}
	\]

\section{Proof of Theorem \ref{thm-stab-reults}}\label{proof of thm-stab-reults}
	The linearized system at a given stationary state $E^*=(B^*,u^*,v^*)$ writes 
	\begin{equation}\label{eq-GenModel-Linear}
	\left(\dot B,\dot u,\dot v\right)^T=J[E^*] \left(B,u,v\right)^T
	\end{equation}
	wherein $J[E^*]$ is defined by the Jacobian matrix associated to \eqref{eq-GenModel-Linear} and is given by 
	\begin{equation}\label{eq-JacobMatrix}
	\begin{split}
	& J[E^*]= \left(\begin{array}{ccc}
	-\left[d + \left<\beta,(u^*,v^*)\right>\right] & - B^* (\beta_s,\beta_m) & - B^* (\beta_p,\beta_{m.p})\\
	Gu^*+L_pv^* & B^*G-D- M[E^*] & -N[E^*] +B^*L_p\\
	(G_p-L_p)v^* & M[E^*] & B^*(G_p-L_p)-D_p +N[E^*]
	\end{array} \right),\\
	& M[E^*]= \left(\begin{array}{cc}
	H(N^*)(N^*_p+N^*_{m.p}) + \xi^*_s  & \xi^*_s  \\
	\xi^*_m & H(N^*)(N^*_p+N^*_{m.p}) + \xi^*_m
	\end{array} \right),\\
	& N[E^*]= \left(\begin{array}{cc}
	H(N^*)N^*_s+ \xi^*_s &H(N^*)N^*_s+ \xi^*_s \\
	H(N^*)N^*_m+ \xi^*_m & H(N^*)N^*_m + \xi^*_m
	\end{array} \right),
	\end{split}
	\end{equation}
	with $\xi^*_s=H'(N^*)N^*_s(N^*_p+N^*_{m.p}),$ and $\xi^*_m=H'(N^*)N^*_m(N^*_p+N^*_{m.p})$. Without loss of generality, and when necessary, we express parameters $\varepsilon_j$ and $\alpha$ as functions of the same quantity, let us say $\eta$, with $\eta\ll 1$.
	
	Recall that the stability modulus of a matrix $M$ is $s_0(M)=\{\max \text{Re}(z): z\in \sigma(M) \}$ and $M$ is said to be locally asymptotically stable (l.a.s.) if $s_0(M)<0$ \cite{LiWang1998}. 
	
	
	\paragraph{Stablity of $E^*_{s-m}$.} At point $E^*_{s-m}$, the Jacobian matrix $J[E^*_{s-m}]$ writes 
	\[
	J[E^*_{s-m}]= \left(\begin{array}{cc}
	Z[E^*_{s-m}] & X[E^*_{s-m}]\\
	0 &  Y[E^*_{s-m}]
	\end{array} \right),
	\]
	with 
	\[
	\begin{split}
	& Z[E^*_{s-m}]= \left(\begin{array}{cc}
	-\left[d + \left<(\beta_s,\beta_m),u^*\right>\right] & - B^* (\beta_s,\beta_m) \\
	Gu^* & B^*G-D 
	\end{array} \right),\\
	& X[E^*_{s-m}]= \left(\begin{array}{c}
	- B^* (\beta_p,\beta_{m.p})\\
	-N[E^*_{s-m}] +B^*L_p
	\end{array} \right),\\
	& Y[E^*_{s-m}]=  B^*(G_p-L_p)-D_p +N[E^*_{s-m}].
	\end{split}
	\]

	Then, $E^*_{s-m}$ is unstable if $s_0 \left( Y[E^*_{s-m}]\right)>0$ and a necessary condition for the stability of $E^*_{s-m}$ is that $s_0 \left( Y[E^*_{s-m}]\right)<0$. By setting $h^*_j= \frac{N^*_j}{a_H+b_HN^*}$, with $j=s,m$; note that 
	\[
	Y[E^*_{s-m}]= 
	\left[
	\begin{array}{cc}
	B^*\tau_p\beta_p(1-\varepsilon_p)(1-\theta) -d_p +\alpha h^*_s&  B^*\varepsilon_{m.p}(1-\theta)\tau_{m.p}\beta_{m.p} +\alpha h^*_s \\
	B^*\varepsilon_p(1-\theta)\tau_p\beta_p +\alpha h^*_m &  B^*\tau_{m.p}\beta_{m.p} (1-\varepsilon_{m.p})(1-\theta) -d_{m.p} +\alpha h^*_m
	\end{array}
	\right]
	\]
	and eigenvalues $z_1$, $z_2$ of $Y[E^*_{s-m}]$ are such that 
	\[
	\begin{split}
	z_1=& B^*\tau_p\beta_p(1-\varepsilon_p)(1-\theta) -d_p +\alpha h^*_s + \mathcal{O}(\eta^2)\\
	=& d_p \left( \frac{B^*}{B_0}\mathcal{R}_p +\frac{\alpha h^*_s}{d_p} -1\right)  + \mathcal{O}(\eta^2),\\
	z_2=& B^*\tau_{m.p}\beta_{m.p} (1-\varepsilon_{m.p})(1-\theta) -d_{m.p} +\alpha h^*_m + \mathcal{O}(\eta^2)\\
	=&d_{m.p} \left( \frac{B^*}{B_0}\mathcal{R}_{m.p} +\frac{\alpha h^*_m}{d_{m.p}} -1\right) +\mathcal{O}(\eta^2).
	\end{split}
	\]
	
	From where, for sufficiently small mutation and flux of HGT rates, we have 
	\[
	s_0(Y[E^*_{s-m}])<0 \Longleftrightarrow 
	\left\{
	\begin{split}
	& \frac{\mathcal{R}_p}{\max \left(\mathcal{R}_s,\mathcal{R}_m\right)} +\alpha \frac{ h^*_s}{d_p} <1,\\
	& \frac{\mathcal{R}_{m.p}}{\max \left(\mathcal{R}_s,\mathcal{R}_m\right)} + \alpha \frac{ h^*_m}{d_{m.p}} <1.
	\end{split}
	\right.
	\]
	
	Next, it remains to check the stability of the block matrix $Z[E^*_{s-m}]$: 
	\[
	Z[E^*_{s-m}]= \left(\begin{array}{ccc}
	-\left[d + \beta_sN^*_s+ \beta_m N^*_m\right] & - B^* \beta_s & -B^*\beta_m \\
	\tau_s\beta_s(1-\varepsilon_s)N^*_s +\varepsilon_m\tau_m\beta_mN^*_m & B^*\tau_s\beta_s(1-\varepsilon_s) -d_s & B^*\varepsilon_m\tau_m\beta_m\\
	\varepsilon_s\tau_s\beta_s N^*_s+ \tau_m\beta_m(1-\varepsilon_m)N^*_m & B^*\varepsilon_s\tau_s\beta_s & B^*\tau_m\beta_m(1-\varepsilon_m) -d_m
	\end{array} \right).
	\]
Again, recalling that $(N^*_s,N^*_m)= \frac{d\left(B_0r(D^{-1}G)-1\right)}{\left< (\beta_s,\beta_m),\phi_0\right>} \phi_0 $, the expansion of the matrix $Z[E^*_{s-m}]$ takes the form $	Z[E^*_{s-m}]= 	Z^0+ \eta Z^\eta$, wherein 
\[
	Z^0= \left(\begin{array}{ccc}
	-\left[d + \beta_s n_s^0 \right] & - b^0 \beta_s & -b^0\beta_m \\
	\tau_s\beta_sn^0_s  & b^0\tau_s\beta_s -d_s & 0\\
	 0 & 0 & b^0\tau_m\beta_m -d_m
	\end{array} \right).
	\]
with $n^0_s= \frac{c^0}{2} \left( \mathcal{T}_s- \mathcal{T}_m +|\mathcal{T}_s-\mathcal{T}_m | \right)$, $\frac{1}{b^0}= \frac{1}{2B_0} \left( \mathcal{T}_s+ \mathcal{T}_m +|\mathcal{T}_s-\mathcal{T}_m | \right)$, and where $c^0$ is a positive constant (which does not depends on $\eta$). We first assume that $\mathcal{T}_s > \mathcal{T}_m$. Then, eigenvalues $z^0_k$ ($k=1,2,3$) of $Z^0$ are such that $z^0_1= b^0\tau_m\beta_m -d_m$, and $z^0_2$, $z^0_3$ are solution of 
\begin{equation*}
z^2+ \left[d + \beta_s n_s^0+ d_s- b^0\tau_s\beta_s \right]z +b^0\tau_s\beta_s^2n^0_s+ (d + \beta_s n_s^0)( d_s- b^0\tau_s\beta_s)=0.
\end{equation*}
Since $d_s- b^0\tau_s\beta_s= d_s \left(1- \frac{2}{\mathcal{T}_s+ \mathcal{T}_m +|\mathcal{T}_s-\mathcal{T}_m |}   \mathcal{T}_s \right) \ge 0 $, coefficients of the previous polynomial are positive and so $\mathcal{R}_e(z^0_2)<0$, $\mathcal{R}_e(z^0_3)<0$. Moreover, $z^0_1= d_m \left(\frac{2}{\mathcal{T}_s+ \mathcal{T}_m +|\mathcal{T}_s-\mathcal{T}_m |}   \mathcal{T}_m -1 \right) \le 0$, from where $s_0(Z^0)<0$, which gives $s_0(Z[E^*_{s-m}])<0$ for $\eta$ sufficiently small \cite{Kato}. Consequently, for small $\eta$, $E^*_{s-m}$ is l.a.s if and only if  $s_0(Y[E^*_{s-m}])<0$.

\paragraph{Stability of $E^*$.} Let 
\[
\left\{
\begin{split}
&u^*=(x_s^*,x_m^*)+ \mathcal{O}(\eta),\\
&v^*=(x_p^*,x_{m.p}^*)+ \mathcal{O}(\eta),\\
&B^*=  b^*_0+ \mathcal{O}(\eta).
\end{split}
\right.
\]
Then $J[E^*]$ takes the form $J[E^*]= J^*_0 + \eta J^*_\eta$, with $J^*_0= $ 
\begin{equation}
\begin{split}
& \left(\begin{array}{ccccc}
-\left(d + \sum_j \beta_jx^*_j\right) & - b^*_0 \beta_s & - b^*_0\beta_m & - b^*_0 \beta_p & - b^*_0 \beta_{m.p}\\
\mu_s x^*_s +\theta \mu_p x^*_p& b^*_0\mu_s-d_s & 0 &b^*_0\theta \mu_p &0 \\
\mu_m x^*_m +\theta \mu_{m.p} x^*_{m.p} & 0 & b^*_0\mu_m -d_m & 0 & b^*_0\theta \mu_{m.p}\\
(1-\theta) \mu_p x^*_p & 0 &0 &b^*_0(1-\theta)\mu_p -d_p& 0\\
(1-\theta) \mu_{m.p} x^*_{m.p} & 0 &0 &0 & b^*_0(1-\theta)\mu_{m.p}-d_{m.p}
\end{array} \right).
\end{split}
\end{equation}
Recall that, for $\eta$ sufficiently small, the existence of $E^*$ is ensured by $$(1-\theta)\max \left(\mathcal{T}_p,\mathcal{T}_{m.p}\right) > \max \left(\mathcal{T}_s,\mathcal{T}_{m}\right).$$

If $\mathcal{T}_p> \mathcal{T}_{m.p}$, from \eqref{esti-Estar-Rp}, we have $b^*_0=\frac{B_0}{(1-\theta)\mathcal{T}_p}$, $x^*_m=x^*_{m.p}=0$, and $J^*_0$ rewrites 
$J^*_0= $ 
\begin{equation*}
\begin{split}
& \left(\begin{array}{ccccc}
-\left(d + \sum_j \beta_jx^*_j\right) & - b^*_0 \beta_s & - b^*_0\beta_m & - b^*_0 \beta_p & - b^*_0 \beta_{m.p}\\
\mu_s x^*_s +\theta \mu_p x^*_p& b^*_0\mu_s-d_s & 0 &b^*_0\theta \mu_p &0 \\
0 & 0 & b^*_0\mu_m -d_m & 0 & b^*_0\theta \mu_{m.p}\\
(1-\theta) \mu_p x^*_p & 0 &0 &0& 0\\
0 & 0 &0 &0 & b^*_0(1-\theta)\mu_{m.p}-d_{m.p}
\end{array} \right).
\end{split}
\end{equation*}
Denoting by $\sigma(\cdot)$ the spectrum of a given matrix, we have $\sigma(J^*_0)=\left\{b^*_0(1-\theta)\mu_{m.p}-d_{m.p}, b^*_0\mu_m-d_m \right\} \cup \sigma(Y^*_0)$, where 
\begin{equation*}
Y^*_0=\left(\begin{array}{ccc}
-\left(d + \sum_j \beta_jx^*_j\right) & - b^*_0 \beta_s & - b^*_0 \beta_p \\
\mu_s x^*_s +\theta \mu_p x^*_p& b^*_0\mu_s-d_s &b^*_0\theta \mu_p \\
(1-\theta) \mu_p x^*_p & 0 &0
\end{array} \right).
\end{equation*}
Since $b^*_0(1-\theta)\mu_{m.p}-d_{m.p}= d_{m.p} \left(\mathcal{T}_{m.p}/\mathcal{T}_p-1\right)<0 $ and $b^*_0\mu_{m}-d_m= d_{m} \left(\mathcal{T}_{m}/((1-\theta)\mathcal{T}_p)-1\right)<0$, then $E^*$ is l.a.s. iff $s_0\left(Y^*_0\right)<0$. The characteristic polynomial of $Y^*_0$ writes $\pi(z)=z^3 +\pi_2 z^2 + \pi_1 z +\pi_0$, with 
\[
\begin{split}
& \pi_2= d +  \beta_sx^*_s + \beta_px^*_p+ d_s-b^*_0\mu_s,\\
&\pi_1= (d +  \beta_sx^*_s + \beta_px^*_p)(d_s-b^*_0\mu_s) +b_0^*\beta_s (\mu_s x^*_s +\theta \mu_p x^*_p) +b_0^*\beta_p(1-\theta) \mu_p x^*_p,\\
&\pi_0= (1-\theta) \mu_p x^*_p \left((b_0^*)^2\beta_s\theta\mu_p + (d_s-b^*_0\mu_s)b_0^*\beta_p \right).
\end{split}
\]
By the Routh–Hurwitz stability criterion we deduce that $\pi$ is a stable polynomial ({\it i.e.} $s_0\left(Y^*_0\right)<0$), and then $E^*$ is l.a.s..

If $\mathcal{T}_p< \mathcal{T}_{m.p}$, again from \eqref{esti-Estar-Rmp}, we have $b^*_0=\frac{B_0}{(1-\theta)\mathcal{T}_{m.p}}$, $x^*_s=x^*_{p}=0$, and the characteristic polynomial of $J^*_0$ writes 
$|J^*_0 -z \mathbb{I}|= (b^*_0\mu_s-d_s-z ) (b^*_0(1-\theta)\mu_p -d_p-z) \times$ 
\begin{equation*}
\left|\begin{array}{ccccc}
-\left(d + \sum_j \beta_jx^*_j\right)-z & - b^*_0\beta_m &  - b^*_0 \beta_{m.p}\\
\mu_m x^*_m +\theta \mu_{m.p} x^*_{m.p}  & b^*_0\mu_m -d_m -z  & b^*_0\theta \mu_{m.p}\\
(1-\theta) \mu_{m.p} x^*_{m.p}  &0 & -z
\end{array} \right|.
\end{equation*}
Same arguments as previously lead to the local stability of $E^*$.


\section{The model with mutations depending on the abundance of the parental cells}
When the occurrence of new mutants depends on the abundance of the parental cells, the model writes 
\begin{equation}\label{eq-GenModel2}
\left\{
\begin{aligned}
\dot B(t)=& \Lambda-dB -B\sum_{j \in \mathcal{J}}\beta_jN_j ,\\
\dot N_s(t)=&\tau_s\beta_s BN_s  -(d_s+\varepsilon_s)N_s +\varepsilon_mN_m + \theta\tau_p\beta_p BN_p    - H(N) N_s[N_p +N_{m.p}] ,\\
\dot N_m(t)=& \tau_m\beta_m  BN_m  - (d_m+\varepsilon_m)N_m + \varepsilon_s N_s +\theta\tau_{m.p}\beta_{m.p} BN_{m.p}    - H(N) N_m[N_{m.p}+N_p],\\
\dot N_p(t)=& \tau_p\beta_p(1-\theta) BN_p  - (d_p+\varepsilon_p)N_p +\varepsilon_{m.p} N_{m.p}  +H(N) N_s[N_p +N_{m.p}] ,\\
\dot N_{m.p}(t)=& \tau_{m.p}\beta_{m.p}(1-\theta) BN_{m.p}  - (d_{m.p}+\varepsilon_{m.p}) N_{m.p}  + \varepsilon_p N_p + H(N) N_m[N_{m.p}+N_p],
\end{aligned}
\right.
\end{equation} 
wherein state variables and model parameters are the same as for Model \eqref{eq-GenModel}.

\end{appendix}

\end{document}